\documentclass[11pt,fleqn,reqno]{amsart}
\usepackage[parfill]{parskip}
\usepackage{amssymb, amsmath}
\usepackage[titletoc,title,page]{appendix}
\usepackage{graphicx}
\usepackage{float}
\usepackage[round]{natbib}
\usepackage{setspace}
\usepackage{thmtools}
\usepackage[textsize=footnotesize, color=orange!65]{todonotes}
\usepackage{enumerate}
\usepackage{amsthm}
\usepackage{hyperref}
\usepackage{enumitem}
\usepackage{xcolor}

\usepackage[framemethod=default]{mdframed}

\global\mdfdefinestyle{exampledefault}{%
linecolor=lightgray,linewidth=1pt,%
leftmargin=1cm,rightmargin=1cm,
}

\graphicspath{ {../Mainly} }

\let\IG\iffalse

\batchmode

\newcommand{\expect}{\mathbb{E}}

\newcommand{\given}{{\mbox{{\large $\mid$}}}}

\newcommand{\argmin}{\operatornamewithlimits{arg min}}

\newcommand{\bpf}{\begin{singlespace} \noindent {\bf Proof}:\ \ }
\newcommand{\epf}{$\;\rule{1.5mm}{3mm}$ \end{singlespace} \smallskip}

\newcommand{\bphi}{\boldsymbol{\phi}}

\newcommand{\Amb}{\mbox{{\sf $\kappa$}}\,}

\newtheorem{thm}{Theorem}
\newtheorem{prop}{Proposition}

\newtheorem{lemma}{Lemma}

\newtheorem{definition}{Definition}

\def \upbar{\overline}
\def \til[#1]{\widetilde{#1}}

\numberwithin{equation}{section}

\newcommand{\actions}{\mathcal{A}}

\definecolor{lightsalmonpink}{rgb}{1.0, 0.6, 0.6}
\definecolor{lightsalmon}{rgb}{1.0, 0.63, 0.48}
\definecolor{lightpastelpurple}{rgb}{0.69, 0.61, 0.85}

 \definecolor{cornflowerblue}{rgb}{0.39, 0.58, 0.93}
\definecolor{babyblueeyes}{rgb}{0.63, 0.79, 0.95}
\definecolor{babyblue}{rgb}{0.54, 0.81, 0.94}
\definecolor{azure(colorwheel)}{rgb}{0.0, 0.5, 1.0}
\definecolor{cyan(process)}{rgb}{0.0, 0.72, 0.92}
\definecolor{darkcoral}{rgb}{0.8, 0.36, 0.27}
\definecolor{deepcarrotorange}{rgb}{0.91, 0.41, 0.17}
\definecolor{lightseagreen}{rgb}{0.13, 0.7, 0.67}
\definecolor{ao}{rgb}{0.0, 0.0, 1.0}

\let\IG\iffalse

\begin{document}

\title{Moral Hazard, Dynamic Incentives, \\ and Ambiguous Perceptions}

\author{Martin Dumav$^{\ast}$}
\date{October 23rd, 2021 \\
$^{\ast}$\text{Universidad Carlos III de Madrid, Department of Economics}
\\
Many thanks to Eduardo Faingold, William Fuchs, Takashi Hayashi, Urmee Khan, David Levine, Antoine Loeper, Ramon Marimon, Frank Riedel, Johannes Schneider and Gordan Zitkovic for help with this paper.
Earlier elements of this research
originated during my stay at Bielefeld University as a postdoctoral research fellow and the European University Institute as a Max Weber postdoctoral fellow: the hospitality of these institutions are gratefully acknowledged.
Support from the Ministerio Economia y Competitividad (Spain) through grants ECO2017-86261-P and MDM 2014-0431
is gratefully acknowledged.}

\begin{abstract}
This paper considers dynamic moral hazard settings, in which
the consequences of the agent's actions are not precisely understood.
In a new continuous-time moral hazard model with drift ambiguity,
the agent's unobservable action translates to drift set that describe the evolution of output.
The agent and the principal have  imprecise information about the technology,
and both seek robust performance from a contract in relation to their respective worst-case scenarios.
We show that the optimal long-term contract
aligns the parties' pessimistic expectations and broadly features compressing
of the high-powered incentives.
Methodologically, we provide a tractable way to formulate and characterize optimal long-run contracts with drift ambiguity.
Substantively, our results provide some
insights into the formal link between robustness and simplicity of dynamic contracts,
in particular high-powered incentives become less effective
in the presence of ambiguity.

\medskip


\medskip

\end{abstract}

\maketitle

\newpage

\section{Introduction}

This paper studies dynamic agency problems, in which the worker's action translates to outcomes
in ways that are not precisely understood.
These situations are common, especially for white collar employees
in large organizations, constituting a large portion of jobs people work in. Lack of specific responsibilities on the part of
these workers for either generating sales or  the overall performance of a major organizational unit,
such as a division or the entire firm,  limits the availability of
appropriate performance measures sufficient to support
exclusive reliance on explicit financial incentives.
Lack of apparent sensitivity of compensation to performance is a common  occurrence.\footnote{
See, for instance, \citet{holmstrom2017pay}, especially Section II.C.}
 In such settings, it is plausible to imagine that the parties
understand the technological possibilities only to some degree, but not precisely.
What kind of a contract would the principal offer?

The main contribution of this paper is  twofold. Formally, it gives a justification
for incentive schemes that do not vary sensitively with performance in an environment where
the parties' ambiguous perception of consequences of the actions lead them to seek robustness  as in  worst-case guaranteed value.
In the process, the paper also provides new insights into a more general
question about incentive contracting, namely, what happens to ``high-powered incentives'' as
the information available about technology becomes more imprecise in comparison to standard
Bayesian models? Our results suggest that broadly such incentives lose their impact.

In the model, the risk-neutral principal and the risk-averse agent engage in an agency relationship.
They both know the common set of actions available, but for any action, the consequences are ambiguous.
Namely, consequences are perceived imprecisely as
drift ambiguity, or \emph{sets} of distributions over output.
The special case where sets collapse to singletons reduce to Bayesian formulations
of dynamic moral hazard problems with precise information about technology, for instance in
\citet{sannikov2008continuous}'s classic framework
where each action's consequences are precisely understood as a unique probability distribution over outcomes.
The novel modeling approach in the current paper extends the dynamic moral hazard framework by allowing for imprecise information about technology, where each
action induces a \emph{set} of distributions, modeled with drift ambiguity a la \citet{chen_epstein_ambiguity2002}.
The parties have  imprecise understanding of technology,
and perceive drift-ambiguity associated with each action.

The key challenge is as follows: for any given
contract, due to imprecise information, the principal and the agent
can differ in their
perceptions of worst-case scenarios.
This divergence in worst-case perceptions
implies non-separability between the agency cost of implementing an action and its profitability,
unlike the classic moral hazard problem (for instance in \citet{grossman1983analysis}).
It is difficult to
analyze the optimal contract as a result \citep[pg.290]{mukerji2004overview}.
In particular,
in the analysis of dynamic moral hazard problems
recursive structures in the classic Bayesian mould as in \citet{sannikov2008continuous}
 are not applicable.
By systematically incorporating
drift ambiguity a la \citet{chen_epstein_ambiguity2002}
 in a dynamic moral hazard model the current paper presents a tractable approach to the problem of
ambiguous perceptions in the agency relationship and
characterizes the optimal contract in two steps.


In the first-step, for any arbitrary contract the agent's value is
represented recursively using his continuation value $W$
as a state variable. Analogous to  \citet{sannikov2008continuous}
it yields a quantity $Y$ that represents the sensitivity of the agent's
future payments to performance uncertainty.
As a novel element, concern for imprecision about technology, say with strength $\kappa$ that parameterizes the size of the
drift set, adds one additive term $\kappa Y$ relative to the agent's value representation
in a Bayesian model with a unique prior.
This added term $\kappa Y$ discounts the agent's
expected value according to his worst-case scenario due to ambiguity aversion.
Since the sensitivity $Y$ reflects the expected continuation value,
the presence of ambiguity induces a preference for the certainty of payments today
over the expected payments in the uncertain future. Back-loading payments that are effective
in providing incentives in the Bayesian model entails larger agency costs of implementation when there is drift-ambiguity.

In the second step,
we provide a recursive representation for the principal's contracting problem
which facilitates the characterization of the optimal contract.
Assuming stationary drift-sets, the principal's problem
subject to the agent's incentive constraints
reduces to an optimal stochastic control problem with a single variable,
the agent's continuation value $W$.
The approach in this paper heuristically sets up a recursive functional equation,
Hamiltonian-Jacobian-Bellman-Isaacson (HJBI), which extends  the familiar HJB formulation
of the contracting problem and incorporates drift ambiguity.
A novel result establishes that such a heuristic description formally represents
the contracting problem via an extension of verification theorem appropriate for drift ambiguity.

With the formulation of non-probabilistic beliefs, and robustness concerns of the parties,
the analysis shows that presence of imprecision compresses the incentives in the optimal long-run contract.
The driving force behind our results is the observation that imprecision in technology
reduces the expectations and also contributes to the agency costs of long-run incentives, and
the principal's optimal contract design features lower incentives as compared to the Bayesian formulation.
Formalization of this results flows from a novel application of
monotone comparative statics inspired by
\citet{quah2013discounting}'s analysis of optimal control decisions.
Complementing the result on lack of incentive sensitivities,
consumption profile is smooth under ambiguity as the contract is
less sensitive to performance. Furthermore, effort profile and wage scheme are  compressed.

The analysis of the HJBI shows that \emph{during} the employment relationship the principal and the agent align their worst-cases as the lower envelopes of drift sets.
Agreement on the worst-case is an \emph{endogenous} property of
the optimal contract, which takes into account the values outside the employment.
Given that outside the employment there is no ambiguity,
under the lower envelopes as the common worst-case scenarios,
the optimal contract becomes observationally equivalent to
the model with the lower-envelope as the specification of technology.

\subsection*{Related Literature}
Methodologically, this paper
offers a flexible framework to formulate and characterize optimal long-run contracts
in dynamic moral hazard relationships in continuous-time by extending
the classical Bayesian formulation as in
\citet{sannikov2008continuous} and incorporating drift ambiguity
as in\citet{chen_epstein_ambiguity2002}.
Substantively, our results provide some insights into the formal
link between robustness and dynamic contracts.
These qualitative features of the optimal contract
flow from extending and applying  \citet{quah2013discounting}'s analysis of the time-preferences on
optimal control decisions in an appropriate manner for the contracting problem analyzed in the
current paper. Specific relations to these contributions have been already noted.

The main theoretical development in this paper derives a
recursive characterization of the optimal dynamic contract problem
with drift ambiguity.
Earlier contributions in dynamic agency problems
fruitfully utilize formulations based on precise Bayesian information
in various moral hazard settings.%
\footnote{
 The prominent examples of this approach includes:
\citet{cvitanic2009optimal},
\citet{williams2009dynamic},
\citet{williams2011persistent},
\citet{demarzo2013learning},
 \citet{kapivcka2013efficient},
 \citet{garrett2015dynamic} and \citet{sannikov2014moral}.%
}

A few papers have looked at what happens when the agent's action is ambiguously perceived.
The closest related paper is  \citet{miao2013robust} who also consider
Sannikov-style model of moral hazard
with ambiguity about the drift.
In their model of one-sided ambiguity,
the principal is ambiguity averse,
while the agent is ambiguity neutral,
that is, the agent evaluates the contract under
the `true' reference probability measure.
In the current paper, both the principal
and the agent are ambiguity averse and,
potentially, evaluate the contract
under different worst-case measures.
Moreover, in the current setting the set of measures
may depend on the actions that the agent takes,
which is not possible in Miao and Rivera's model.
Therefore, in their model ambiguity does not affect
incentive provision.
In the current paper the main focus
on the interaction of imprecise information about
technology and the structure of incentives is best
reflected in a model with two-sided ambiguity.

Work by \citet{szydlowski2012ambiguity}
introduced ambiguity into a dynamic contracting problem in continuous time.
In that model, the principal is ambiguous about the agent's effort cost.
The formulation in the current paper instead
models
information imprecision in
the decision-theoretic theme of ambiguity, in particular using
drift ambiguity.
 \citeauthor{szydlowski2012ambiguity} shows that the agent receives
excessively strong incentives in a dynamic contracting setting when the principal is
ambiguous about the agentâ€™s cost of effort.
\citet{wu2017ambiguity} extend \citet{holmstrom1987aggregation}'s model
by introducing belief distortions a la Sargent and Hansen.
In such a model they find increase in pay-for-performance sensitivity
and shed light on compensations schemes that reward luck.
Complementing those contributions, the results here
show that broadly high-powered incentives become less
effective when there is drift-ambiguity.

Compressed wages arise
 in the literature
 in different model settings.
In a relational contract setting, \citet{macleod2003optimal}
relaxes the assumption of common knowledge of output between the parties.
\citet{macleod2003optimal} shows that resulting disagreement is best resolved by flattening
the wage profile in the optimal long term contract.
\citet{fuchs2007contracting} shows that such a compressing of wages is robust to
allowing the principal's subjective evaluation to be her private information.
Relative to these contributions, the current paper provides a complementary
rationale for compressed wages in an environment where output is contractible
but the parties have imprecise beliefs about technology.


Our paper fits more broadly in a small literature
on moral hazard under ambiguity in a static model.
A desire for robust contracts
often leads to the use of simple contracts in this literature,
including contributions \citet{hurwicz1978incentive},  and more recently,
\citet{lopomo2011knightian}, \citet{chassang2013calibrated}, \citet{antic2014contracting}, \citet{garrett2014robustness}, \citet{carroll2015robustness},
 and \citet{dumav2017moral,dumav2018implementible}.
These papers provide foundations for contracts in simple forms.
e.g. linear (\citet{carroll2015robustness}, \citet{dumav2017moral}) or step functions (\citet{lopomo2011knightian}).
Robustness of high powered non-linear incentives have been analyzed in specific applications
by \citet{ghirardato1994agency} and \citet{mukerji1998ambiguity}.
The question under what conditions on imprecision
about the actions rationalizes simple forms of contracts in dynamic setting
is left to future research.


\section{A model of dynamic moral hazard with drift ambiguity}\label{section:contracting_problem}
Consider an agency relation taking place in continuous time.
The agent's unobservable action translates to a flow of outputs in a stochastic manner according to
drift ambiguity.
In particular, for any action $a_t$ its drift set $\Theta({a_t})$
describes the \emph{set} of possible drifts for this action, or expected increments
in output diffusion.
Given a sequence of actions $(a_{\tau})_{\tau \leq t}$ up to time $t$
the total output $X_{t}$ evolves according to a \emph{family} of diffusion processes whose members are:
\begin{equation}\label{eqn:Output-AmbigDiffusion}
dX_{t} = ( a_{t} + \theta_{t} ) dt + \sigma dB_{t},
\end{equation}
where the expected increment $\theta_t$ belongs to the drift set $\Theta({a_t})$,
and the noise process $\{ B_{t} \}$ is a  Brownian motion adapted to the filtration $\mathcal{F}$
on a standard filtered probability space
$\left( \Omega,P,\mathcal{F} = \{ \mathcal{F}_{t}; 0 \leq t < \infty \}\right)$.
We model the drift sets as intervals centered at zero:
for any $a\in \mathcal{A}$, the drift set of the action
is $\Theta(a)= [-\kappa(a),\kappa(a)]$
where $\kappa(a)$ describes the size of the drift set of the action $a$.
For simplicity, we assume that $0 < \kappa(a) < a $
and interpret $\kappa$ as strength of ambiguity.
Indeed, taking $\kappa(a) = 0$ for all $a \in \mathcal{A}$
restricts the drifts sets to be  singletons
collapses into the precise information
as in the classical model of dynamic moral hazard
studied by \cite{sannikov2008continuous}.
Unlike in \cite{sannikov2008continuous}'s model,
in the present setting,
the agent's unobservable actions controls
the drift sets.
The productivity $a_{t} + \theta_t$ of the action $a_t$ is therefore not precisely known; rather,
it can be any one of the elements in the set $a_{t} + \Theta({a_t})$.
An interpretation of the specification of technology with drift-ambiguity
is that the parties to the contract are aware of the possibility
that they have erroneous beliefs about the true drift for actions
and seek robust performance.

We assume that the drift sets
$ \Theta({a_t}) = [-\kappa(a_{t}), \kappa(a_{t})]$
are stationary and independent (using
terminology introduced by \citet{chen_epstein_ambiguity2002}).
The drift sets are stationary as they do not depend on time $t$
beyond the action $a_{t}$ at that time,
and independent of history of output realizations
up to that time.

In order to specify the contracting problem classically in terms of expected utilities,
following \citet{chen_epstein_ambiguity2002} we translate drift sets into sets of distributions over outputs.
To illustrate this translation, for simplicity consider a given
finite horizon $T$.
For an action profile $a = \{a_{t}\}$,
we consider, with a slight abuse of notation, a drift process $\theta:=\{\theta_{t}\}$
such that $\theta_{t} \in \Theta(a_{t})$.
Using the drift process
we define a process of Girsanov exponents as follows
\begin{equation}
\label{eqn:girsanov-exponent-model}
z_{t}^{a + \theta} :=
\exp\left\{
-\frac{1}{2} \int_{0}^{t} |a_{s} + \theta_{s}|^{2}ds - \int_{0}^{t} (a_{s} + \theta_{s}) dB_{s}
\right\},
\quad
0 \leq t \leq T,
\end{equation}
where $\theta_{t} \in \Theta(a_{t}) = [-\kappa(a_{t}), \kappa(a_{t})]$.
The Girsanov exponent $\{z_{t}^{a + \theta}\}$ generate
a probability measure
$P^{a+\theta}$ on the measurable space of sample paths $(\Omega, \mathcal{F})$
and the Radon-Nikodym derivative of this measure
with respect to the reference measure $P$ is given by
\begin{equation}
\label{eqn:Radon-Nik-P-model}
\frac{d P^{a+\theta}}{d P}
\bigg\rvert_{\mathcal{F}_{t}} = z_{t}^{a + \theta},
\quad
0 \leq t \leq T.
\end{equation}

For the action $a = \{a_{t}\}$ and the drift process
$\{\theta_{t}\}$ such that $\theta_{t} \in \Theta(a_{t})$
we have therefore constructed a measure $P^{a+\theta}$  equivalent to $P$
via a change of measure. Now,
taking the collection of all such measures $P^{a+\theta}$
yields the set of distributions for the action  $a = \{a_{t}\}$
and it is given by:
\begin{equation}
\begin{split}
\label{eqn:multiprior-a}
\mathcal{P}^{\Theta(a)} :=
\{
P^{a+\theta}: \theta= \{\theta_{t} \in \Theta(a_{t}), 0 \leq t \leq T\}
\text{ and } P^{a+\theta}
\text{ is defined by }
\eqref{eqn:Radon-Nik-P-model}
\}
\end{split}
\end{equation}

This construction of the set of distributions
follows the formulation in \citet{chen_epstein_ambiguity2002} in a single-agent decision setting
and allows the drift sets to depend on the actions (In Appendix, we provide the technical details on the construction of the drift sets  and the sets of distributions).
We have therefore represented drift sets as  sets of probabilities over outputs,
as in the classic theme of ambiguity.

We next study the benchmark model
where the agent and the principal have the symmetric perception of imprecision:
they both perceive $\Theta({a})$ to be the possible set of drifts associated to the action $a$.
The assumption of common symmetrically drift-ambiguity perceptions
allows for tractability and heterogenous perceptions are treated afterwards as an extension.
Given ambiguity perceptions, we assume  MaxMin criteria for ambiguity-sensitive preferences.
Next we specify the contracting problem in this environment with drift-ambiguity.

\subsection{Contracting problem with drift ambiguity}

The principal's contract  to the agent specifies a stream of wage payments
$\{c_{t}\}$ and an incentive-compatible recommendation of action $\{a_{t}\}$.
The action $a:= \{ a_{t} \in \mathcal{A}, 0 \leq t < \infty \}$
is a measurable function with respect to the filtration $\mathcal{F}$ generated by the standard Brownian motion and the set of possible effort levels $\mathcal{A}$ is a compact subset of $R_{+}$ with zero as the smallest element.
The agent's disutility from effort level $a \in \mathcal{A}$
is given by $h(a)$ measured in terms of utility of consumption, and
the function $h$ is continuous, increasing and convex.
Furthermore, we assume that $h(0)=0$ as a normalization.

The principal's contract specifies
non-negative flow wage payments
$c := \{ c_{t}; 0 \leq t < \infty  \}$ measurable with respect to the filtration $\mathcal{F}$
The principal can commit to any such long-term contract.
Regarding the agent's preferences, we assume that
his utility function over monetary consequences $u$ is bounded from below and
satisfies a normalization $u(0)=0$.
Moreover, we assume that the utility function
$u : [0,\infty) \rightarrow [0,\infty) $ is an  increasing, strictly concave,
and $C^{2}$ function that satisfies
$u'(c) \rightarrow 0$ as $c \rightarrow \infty$.
These assumptions on the primitives of the economic environment
are adopted following the framework of Sannikov.

For a given contract $(c,a)$ the parties evaluates the contract
according to ambiguity sensitive preferences.
In particular, we assume that the parties' ambiguity sensitive preferences
are represented by MaxMin Expected Utility and hence
 the parties  evaluate the contract
according to their respective worst-case scenarios for drifts.
For simplicity we assume that the principal and the agent
use a common discount rate $r > 0$
on the flow of profit and utility.
From the contract $(c,a)$ the principal's ex-ante guaranteed value is therefore computed as follows
\[
\min_{\theta \in \Theta({a})}
\expect^{a + \theta}
\left[r \int_{0}^{\infty} e^{-rt}( dX_{t} - c_{t}dt ) \right]
=
\min_{\theta \in \Theta({a})}
\expect^{a + \theta}
\left[r \int_{0}^{\infty} e^{-rt}( a_{t} + \theta_{t} - c_{t})dt  \right]
\]
and the agent obtains the ex-ante guaranteed value
\[
\min_{\theta \in \Theta({a})}
    \expect^{a + \theta}
    \left[ r \int_{0}^{\infty} e^{-rt} \left( u(c_{t}) - h(a_{t}) \right)dt \right]
\]
where the output process $dX_{t}$ evolves according to \eqref{eqn:Output-AmbigDiffusion}.
Here we abuse the notation slightly,
and denote by $\Theta(a)$ the drift-sets
associated to the action $a = \{a_{t}\}$,
where a representative element of the set $\Theta(a)$ is
given by $\theta = \{\theta_{t}\}$ such that $\theta_{t} \in \Theta(a_{t})$.

Following the classical definition of incentive-compatibility,  in the current environment
with drift-ambiguity  we say that in the contract $(c,a)$ the action recommendation
$a$ is incentive-compatible if
the action $a$ maximizes the agent's guaranteed value under the payment scheme $c$.
The payment scheme $c$ is incentive-compatible for the action
$a$ as from any alternative action $\hat{a} \not = a$
the agent obtains weakly smaller guaranteed value than
the action $a$
\begin{equation}
\label{eqn:IC-sequential}
\min_{\theta \in \Theta({a})}
    \expect^{a + \theta}
    [ r \int_{0}^{\infty} e^{-rt} \left( u(c_{t}) - h(a_{t}) \right)dt ]
    \geq
\min_{\theta \in \Theta({\hat{a}})}
    \expect^{\hat{a} + \theta}
    [ r \int_{0}^{\infty} e^{-rt} \left( u(c_{t}) - h(\hat{a}_{t}) \right)dt ]
\end{equation}

Letting $\mathcal{A}_{I}$ denote the set of incentive-compatible
payment and recommended action pairs $(c,a)$,
in this setting with drift ambiguity, the principal's incentive-compatible contract offer
maximizes his expected profit under his worst-case criterion
\begin{equation}\label{eqn:PrincipalsMinMax}
\max_{(c,a) \in \mathcal{A}_{I}}\min_{\theta \in \Theta({a})}
\expect^{a + \theta}
\left[r \int_{0}^{\infty} e^{-rt}( a_{t} + \theta_{t} - c_{t})dt  \right]
\tag{P}
\end{equation}
subject to ensuring the agent ex-ante value of at least $\hat{W} \geq 0$
\begin{equation}\label{eqn:PK}
    \min_{\theta \in \Theta({a})}
    \expect^{a + \theta}
    \left[ r \int_{0}^{\infty} e^{-rt} \left( u(c_{t}) - h(a_{t}) \right)dt \right] \geq \hat{W}
\tag{PK}
\end{equation}

We assume and normalize to zero the value of the agent's outside option.
As the wages are non-negative and the agent can always choose zero effort,
in the contracting problem \eqref{eqn:PrincipalsMinMax}
the principal considers only the incentive-compatible
contracts $(c,a)$ that ensure the agent's participation: \eqref{eqn:PK} holds for
$\hat{W} = 0$ in \eqref{eqn:PK}.
Moreover, as the principal always has the option of not hiring the agent,
without loss we focus on the contracting relationship that yields
non-negative guaranteed payoff for the principal.
We refer to the problem in \eqref{eqn:PrincipalsMinMax}
as the sequential contracting problem.

In the contracting problem with drift ambiguity,
even though the principal and the agent has symmetric perception of drift ambiguity,
${\Theta({a})}$ for any action $a$, their worst-case scenarios can differ as the parties have different objectives.

In order to analyze how drift ambiguity affects the contracting relationship,
we introduce a parameterization on  the size of drift sets to formalize
different degrees of ambiguity.
Intuitively, larger drift sets reflect higher degree of ambiguity.
In particular, we introduce a parameter $\phi$
and denote by $\Theta(a;\phi)$ the drift set for the action
under the parameter $\phi$.
The drift sets under the technology $\phi$ is given by
$\Theta(a;\phi) = a + [-\kappa(a;\phi),\kappa(a;\phi)]$ for $a \in \mathcal{A}$.
As the function $\kappa$ measures the size of the drift sets,
we interpret it as a measure of ambiguity.
In particular,
we use this parameterization
to rank drift ambiguity in different economic environments,
and say that the degree
of ambiguity is stronger if the drift sets are larger
Formally, we define this relation as follows
\begin{definition}\label{defn:higher-ambig}
\normalfont
Technology characterized by $\widehat\phi$ is more ambiguous than $\phi$, written as
$\widehat\phi  \succ_{\Amb} \phi$ if for all actions $a \in \actions$,
the  drift sets  satisfy
$\kappa(a;\phi) \leq \kappa(a;\widehat\phi)$.
\end{definition}

This relation provides a partial order over the drift sets as the set containment
in the definition may not hold uniformly over all actions.
With this indexing we introduce a  list of the key assumptions on ambiguity
about technology that streamline
comparative static analysis. These assumptions on the mapping  $\kappa(a;\phi)$
that determines  strength
of ambiguity are: \vspace*{-2mm}
\begin{enumerate}[label=(\normalfont{A}.{\arabic*}),ref=(\normalfont{A}.{\arabic*})]
\item \label{assmp:kappa-A1} Viability of `zero' actions: $\kappa(0;\phi) = 0$; \vspace*{-2mm}
\item \label{assmp:kappa-A2} Better pessimistic expectations from higher effort:
	$\kappa_{a}(a;\phi) < 1$; \vspace*{-2mm}
\item \label{assmp:kappa-A3} Non-trivial role for drift ambiguity: the size of the drift set
$\kappa(a;\phi)$ is not too concave in effort,
	i.e., $-\kappa_{aa} \frac{a}{1 - \kappa_{a}} \leq 1$; \vspace*{-2mm}
\item \label{assmp:kappa-A4} Complementarity between effort and ambiguity: $\kappa_{\phi a} > 0$. \vspace*{-2mm}
\end{enumerate}

\subsection*{Discussion of Assumptions}

Assumption~\ref{assmp:kappa-A1} ensures that the set of admissible actions $\actions_{R}$,
include those actions that generate positive surplus according to the worst-case scenario,
i.e., lower-envelope of the drift set. We call the set of such actions $\actions_{R}$ as viable.

Assumption~\ref{assmp:kappa-A2} is a monotonicity condition describing that
higher actions yield  better pessimistic expectations, more precisely higher lower-envelopes of the drift sets.
This assumption implies that
increasing effort leads to the drift sets that have higher lower envelopes.

Assumption~\ref{assmp:kappa-A3} ensures that
the agent's optimization problem (below) for any arbitrary viable contract
is a concave problem and the first-order conditions are sufficient
to characterize the agent's decision rule for his action choice.

Single-crossing condition \ref{assmp:kappa-A4} in the current context
says that as the strength of ambiguity increases in order to
maintain the expected output effort must increase.
This is analogous to the single-crossing property commonly used in models of contracting
in private information environments. As in those model, ``better'' type of the agent,
here corresponding to lower ambiguity technology, is also marginally better for all margins.

A minimalist example for drift-ambiguity consistent with
the assumptions~\ref{assmp:kappa-A1}-~\ref{assmp:kappa-A4} is the linear model:
$\kappa(a;\phi) = \phi a$ with $\phi \leq 1$. This reflects technology where ambiguity
increase in effort, arguably a worker exerting higher effort confronts
more complex tasks that are harder to measure. This class of examples
is examined in detail below
as well as the complementary specification where higher actions reduce drift ambiguity.

\subsection*{Admissible Contracts}

The description of a contract $(c_{t},a_{t})$
as a pair of wage and action recommendation
allows for two kinds of retirement clauses.
First, it can involve a retirement clause,
also referred to as a `golden parachute,' which specifies a stopping-time $\tau$
such that the agent receives some constant payment $c_{t} = \xi  > 0$
and does not exert effort $a_{t} = 0$ for $t \geq \tau$ if this stopping-time arrives.
For a second kind of retirement, the contract can specify a termination
clause in which the principal effectively fires the agent.
Such a termination can again be represented
with a stopping time $\tau$ such that
$a_{t} = 0$ and $c_{t} = 0$ for $t \geq \tau$.
In either case of retirement with a lump-sum payment $\xi \geq 0$
at the stopping-time $\tau$ measurable with respect to $\mathcal{F}$
can be equivalently expressed as flow payment $\xi$ after $\tau$
using the observation that $\xi = \int_{\tau}^{\infty} r e^{-rt} \xi dt$.

To ensure that the parties' payoffs
in the contracting problem \eqref{eqn:PrincipalsMinMax}
are well-defined, we assume that
the set of admissible contracts
satisfy the integrability
conditions.
For simplicity, we also assume that the set of admissible contracts
has a finite random horizon.
As the action set $\mathcal{A}$ is compact and for any action $a \in \mathcal{A}$
the drift-set $\Theta(a)$ is a compact set,
it is without loss of generality to assume that
the set of admissible contracts $(c,a) = (c_{t},a_{t})$
satisfy the following square integrability conditions%
\footnote{
As the action $\mathcal{A}$ is a compact subset of $R_{+}$
and for each action $a \in \mathcal{A}$ the drift-set $\Theta(a)$ is compact,
the integrability conditions \eqref{eqn:intableA}
and \eqref{eqn:intableA-tau} are sufficient as they imply that
the following integrability holds for any admissible contract $(c,a)$:
\begin{equation*}
\quad
\sup_{a \in \actions } \sup_{\theta \in \Theta(a)}
\expect^{a + \theta}
\left[
\big( e^{-r\tau} \xi \big)^{2}
+
\int_{0}^{\tau} \big( e^{-rt} c_t \big)^{2} dt
\right] < \infty;
\end{equation*}
and the random horizon $\tau$ is finite
\[
\quad
\lim_{n \rightarrow \infty}
\sup_{a \in \actions } \sup_{\theta \in \Theta(a)}
{P}^{a + \theta} [\tau \geq n] = 0
\]
\citet{lin2020random} uses these stronger conditions
as in their more general
formulation of dynamic moral hazard problem
the agent's action controls both the drift and the volatility
of the output process.
In the current model, the agent's action controls
on the drift of output process and hence
the required integrability conditions are simpler.
}

%
%
\begin{equation}\label{eqn:intableA}
\quad
\expect
\left[
\big( e^{-r\tau} \xi \big)^{2}
+
\int_{0}^{\tau} \big( e^{-rt} c_t \big)^{2} dt
\right] < \infty
\end{equation}
and the random horizon $\tau$ is finite
\begin{equation}\label{eqn:intableA-tau}
\quad
\lim_{n \rightarrow \infty}
{P} [\tau \geq n] = 0
\end{equation}
We denote by $\boldsymbol{\mathcal{Z}}(W_{0})$
the set of admissible contracts $(c,a)$ that
satisfy the integrability conditions \eqref{eqn:intableA}
and
\eqref{eqn:intableA-tau}
and
provide the agent with the ex-ante value $W_{0} \geq 0$.


We remark
the assumption that
in an admissible contract
$(c_{t},a_{t})$ the payment scheme $\{c_{t}\}$
and incentive compatible action $\{a_{t}\}$
are measurable with respect to the filtration $\mathcal{F}$ generated by the standard Brownian motion $B$
does not entail loss of generality.
Indeed, given that the payment scheme
$\{c_{t}\}$ is measurable with respect to the filtration $\mathcal{F}$,
the incentive-compatible action $\{a_{t}\}$
determined by \eqref{eqn:IC-sequential}
is also a measurable function with respect to the filtration $\mathcal{F}$.
In this case,
as noted by \citet[Chapter 5]{cvitanic2012contract})
without loss of generality,
for any finite horizon of contracting problem,
an admissible contract can alternatively  be
specified as a function of history of output realizations, as
in \citet{sannikov2008continuous}.%
\footnote{
Many thanks to an anonymous referee for pointing out this clarification
on the different forms of measurability requirements.}

\subsection{Incentive compatibility with drift-ambiguity}\label{section:IC}
An action profile $\{a_{t}\}$ is implementable if there is a contract with output contingent payments
$\{c_{t}\}$ so that the agent chooses this action profile.
We use this usual definition of implementability
while incorporating drift ambiguity.

For a given contract $(c,a)$,
assume that the agent follows the recommended action $a$.
The agent's continuation value $W_{t}$ at time $t$ is defined as
\begin{equation}
\label{eqn:Agent-continuationvalue}
W_{t} = \min_{\theta \in \Theta({a})}
    \expect^{a + \theta}
    [ r \int_{t}^{\tau} e^{-rt} \left( u(c_{t}) - h(a_{t}) \right)dt ]
\end{equation}
where $\tau$ is the random finite stopping-time
associated to the contract $(c,a)$.

To characterize implementability
we start with representing as a diffusion process
the agent's value from any contract and action strategy.
Following \citet{chen_epstein_ambiguity2002}'s representation of
recursive ambiguity sensitive
preferences in continuous-time
under the max-min criterion the agent's payoff can be
represented as a diffusion process:
\begin{prop}
\label{prop:Wdiff-Ambig-A}
\normalfont
For any wage scheme $\{c_{t}\}$ and any action strategy $\{a_{t}\}$ with its associated
set of drift terms
$a_{t} + \Theta(a_t;\phi) = [a_{t} -\kappa(a_t;\phi),a_{t} +\kappa(a_t;\phi)]$
there exists a progressively measurable process $\{Y_{t}\}$ such that
\begin{equation}\label{eqn:Wdiff-Ambig-A}
W_{t}= W_{0} +
 \int_{0}^{t}r \left( W_{s} - u(c_{s}) + h(a_{s}) + \min_{\theta_s \in \Theta(a_s;\phi)} \theta_s |Y_{s}|   \right)ds +
 \int_{0}^{t} r Y_{s}dB^{a}_{s}
\end{equation}.
\end{prop}
In this representation the action strategy  $\{a_{t}\}$ need not be optimal for the
wage scheme $\{c_{t}\}$ contingent on output history.
In particular, it holds for the optimal action strategy.

The novel aspect of this representation due to drift-ambiguity
is the effect that the agent's aversion to ambiguity
introduces one additive term and
discounts his value according to his worst-case scenario
in the drift set: $\min_{\theta_t \in \Theta(a_t;\phi)} \theta_t Y_{t}$.
In \citet{miao2016robust}'s model
of dynamic moral hazard with drift ambiguity,
the agent is ambiguity neutral
and hence in Miao and Rivera's representation result
ambiguity does not play a role.
In our representation using MaxMin preferences to model ambiguity aversion,
the effect of ambiguity disappears if the agent has
precise information about the drift sets, i.e., $\kappa = 0$.

In any viable contracting relationship
the provision of non-trivial incentives (that is needed to implement a positive action, $a_t >0$)
requires a positive sensitivity process $Y_t > 0$.
This in turn implies that the lower envelope is the worst-case scenario perceived by the agent.
Since zero-contract that implements zero action at no cost to the principal is
always feasible, non-trivial incentives in a viable contracting relationship
requires a positive variation process.
During employment
the agent's worst-case scenario for any non-trivial action $a_t > 0$ is then given by
the drift term that minimizes the expected continuation value over the drift set:
$\argmin_{ \{ \theta \in \Theta(a_t;\phi) \} } \theta Y_{t} = \min\{\theta \in \Theta(a_t;\phi) \} = a_t - \kappa(a_t;\phi)$ --
 the lower-envelope of the drift set.
Denoting the agent's worst-case scenario on drift sets by
$ \theta^{A} (a_t;\phi)$
the following result summarizes this observation:
\begin{lemma}
\label{lemma:Alower-envelope}
\normalfont
For any non-trivial contract, the sensitivity $\{Y_t\}$ process is non-negative and
for any optimal action strategy $\{a_{t}\}$ the agent's worst-case scenario is
the lower envelope of drift sets:
$ \theta^{A}(a_t;\phi)= a_t - \kappa(a_t;\phi)$.
\end{lemma}

Using \normalfont{Lemma}~\ref{lemma:Alower-envelope} in \normalfont{Proposition}~\ref{prop:Wdiff-Ambig-A},
from any viable contract $\{c_{t}\}$ and optimal action strategy $\{a_{t}\}$ the agent's value as a diffusion takes the following form:
\begin{equation}\label{eqn:Wdiff-Ambig-A-kappa}
W_{t}= W_{0} +
 \int_{0}^{t}r \left( W_{s} - u(c_{s}) + h(a_{s}) -  \kappa(a_s;\phi) Y_{s}  \right)ds +
 \int_{0}^{t} r Y_{s} dB^{a}_{s}.
\end{equation}

An interpretation of this diffusion representation of the agent's value  up to
the term involving $\kappa$ due to ambiguity aversion is analogous to that in~\citet{sannikov2008continuous}.
As in the latter, after any history of output  up to time $t$
its drift represents the expected increment in the agent's expected value. It involves
accumulating interest coupon on the promised payments $r W_t$,
adding his effort cost $h(a_t)$, and
reducing it by utility of consumption from the payment received today $u(c_t)$.
The last term in the drift is a novel negative term added  due to the drift ambiguity and aversion to ambiguity.
This term reduces the expected continuation value according to the agent's pessimistic expectations
according to his worst-case scenario for the drift
$ \kappa(a_t;\phi) Y_{t} $, which is proportional to
the sensitivity of the agent's continuation-value to output uncertainty
or to the strength of ambiguity.

As Proposition~\ref{prop:Wdiff-Ambig-A} shows relative to the classic case (where $\kappa = 0$) drift ambiguity reduces the agent's value according to his worst-case scenario.
The drift of the agent's value is determined by the allocation of payments over time as in \citet{sannikov2008continuous}.
All else equal, drift ambiguity $\kappa > 0$ and ambiguity aversion reduces the drift on the agent's value diffusion
and hence reflect stronger preference for front-loaded payments relative to the case without ambiguity.
Therefore, the novel effect of ambiguity in contract design is to decrease the benefits of back-loaded payments, captured by $\kappa(a_s;\phi) Y_{s} $.
Since the backloaded payments are important to incentivize the agent's actions in dynamic moral hazard problems,
deferred payments may not be as effective when ambiguity is present.

To analyze the effect of ambiguity on the optimal contract
we next derive incentive compatibility condition
using the representation of
the agent's value~\eqref{eqn:Wdiff-Ambig-A-kappa}
for any contract. The following result
characterizes incentive compatibility in the current setting
\begin{prop}
\label{prop:IC-kappa}
\normalfont
For any action strategy $a=\{a_{t}\}$ and wage scheme $c=\{c_{t}\}$,
let $\{Y_{t}\}$ be the volatility process from Proposition~{\rm \ref{prop:Wdiff-Ambig-A}}.
Then the action strategy $a$ is incentive compatible if and only if
\begin{equation}\label{eqn:IC-kappa-A}
\forall \tilde{a}_t \in \mathcal{A} \quad
Y_{t}(a_{t} - \kappa(a_{t};\phi) ) - h(a_{t})  \geq
Y_{t}(\tilde{a}_{t} - \kappa(\tilde{a}_{t};\phi) ) - h(\tilde{a}_{t})
\ \ dt \otimes dP \ a.e.
\end{equation}
\end{prop}

Notice that setting
$\kappa(a;\phi) =0 $ in condition \eqref{eqn:IC-kappa-A} for all $a \in \mathcal{A}$
specializes it to \citet{sannikov2008continuous}'s
incentive compatibility condition in the classical case without drift ambiguity.
Compared to the latter, drift-ambiguity introduces
one additive term to the incentive compatibility comparison.
These additional terms discounts the agent's continuation value according to the worst-case scenario due to ambiguity aversion.
With a tractable incentive compatibility condition as in \eqref{eqn:IC-kappa-A}
we next turn to analyze how ambiguity
affects incentives needed to implement actions.

\subsection*{On the incentive benefits of back-loaded payments under ambiguity}
The incentive compatibility condition \eqref{eqn:IC-kappa-A}
provides a tractable way to analyze how incentives needed to implement an action
depend on ambiguity.
It implies that given any non-trivial incentives $Y>0$ the agent's decision rule for his optimal action choice
maximizes the agent's expected continuation value net of effort cost
(suppressing the dependence on employment history to preserve notation):
\begin{equation*}
\max_{a \geq 0} (a - \kappa(a;\phi)) Y - h(a)
\end{equation*}
In view of the assumptions \ref{assmp:kappa-A1}-\ref{assmp:kappa-A4},
the objective function in this maximization programme is strictly concave.
Therefore,
the first-order condition is sufficient to characterize the optimal action and it is given by
$(1 - \kappa_{a}(a;\phi)) Y - h'(a) = 0$.
This in turn implies that the minimum variation $Y(a;\phi)$ needed to implement a given action $a > 0$
solves the first-order condition:
\begin{equation}\label{eqn:IC-Y}
Y = \frac{h'(a)}{1- \kappa_{a}(a;\phi)}.
\end{equation}
For any non-trivial effort level $a > 0$
the incentives $Y(a;\phi)$ to implement it increases in the
degree of ambiguity
by the assumed complementarity
(Assumption~\ref{assmp:kappa-A4}) between effort and ambiguity, $\kappa_{a\phi} > 0$.
Moreover, for any degree of ambiguity $\phi$, Assumption~\ref{assmp:kappa-A3} and  \eqref{eqn:IC-Y} together imply that implementing higher effort levels
requires stronger incentives.
These observations are summarized in the following lemma:
\begin{lemma}
\label{lemma:Mon-Y}
\normalfont
Under the assumptions~\ref{assmp:kappa-A1}-\ref{assmp:kappa-A4},
for any non-trivial effort level $a > 0$,
the required incentives to implement it  increases in ambiguity:
$Y(a;\widehat\phi) \geq Y(a;\phi)$ whenever $\widehat\phi  \succ_{\Amb} \phi$.
Moreover, for any given degree of ambiguity
the incentives increase in effort level:
$Y(\widehat{a};\phi) \geq Y(a;\phi)$ whenever $\widehat{a} \geq a $.
\end{lemma}
This result says that
implementing an action requires stronger incentives
as the strength of ambiguity increases.
To illustrate in a special case, consider the $\kappa-$ignorance model
where degree of ambiguity increases linearly in effort level:
$\kappa(a;\phi) = \phi a $ with $\phi > 0$.
To implement any given action $a$
notice from \eqref{eqn:IC-Y} that
the principal needs to provide stronger incentives
under ambiguity:
$Y(a;\phi) = \frac{h'(a)}{1- \phi}  > h'(a) = Y(a;0)$.
In other words,
the incentive constraint become more difficult to satisfy
under ambiguity.
In particular, if the incentives $Y$
become too high, the principal can be better
off by retiring the agent by effectively setting $Y=0$
and implementing zero-action $a=0$.

As the incentive-compatibility condition describes
the agent's decision rule for his action choice, it
depends only on the agent's ambiguity
but not on the principal's ambiguity.
We show later this observation
in an extension of the model where the principal and the agent can perceive different
ambiguity.
In particular, if the principal does not perceive ambiguity,
relative to the \citet{sannikov2008continuous}'s model
the agent's ambiguity require stronger incentives to
satisfy the incentive compatibility condition,
because the agent's worst-case is always the lower envelope of the drift set
for any non-trivial contract.
Now, if the principal perceives ambiguity,
stronger incentives can make a high drift term her worst-case scenario,
as a high drift term corresponds to a distribution
under which a high payment is more likely.
Since the principal's worst-case scenario is \emph{endogenous},
the contracting problem in the current setting
is not equivalent to
a version of
the \citet{sannikov2008continuous}'s model
by changing the drift of each action with
the lower envelope of its drift set.%
\footnote{Many thanks to an anonymous referee for helping to point out
this distinction.}

It is worth remarking that in \citet{miao2016robust}'s model
of dynamic moral hazard with drift ambiguity,
as the agent is ambiguity neutral in their setting,
the agent's incentive compatibility condition does not
depend on ambiguity. Therefore, in Miao and Rivera's model
ambiguity does not affect how difficult it is for the principal
to incentivize high effort from the agent.
In contrast, in the current setting the incentive compatibility
condition depends on drift ambiguity perceived by the agent.
As the incentive compatibility condition describes
the agent's decision rule for his action choice, it
depends on ambiguity perceived by the agent
but not on what ambiguity perception the principal can have.%
\footnote{Many thanks to anonymous referees for highlighting
this novel effect in the current model.}

Having obtained a tractable incentive compatibility condition for implementation
we then turn to characterize the optimal contract
 and study how it depends on degree of ambiguity.

\section{Basic Properties of Optimal Contracts}
\label{section:Basic Properties}
We next represent
the (sequential) contracting problem
\eqref{eqn:PrincipalsMinMax}
recursively using the agent's continuation value as a state variable
and then characterize its properties.


For a given contract $(c,a)$,
the agent's value process is recursively represented by
\eqref{eqn:Wdiff-Ambig-A}
and its diffusion is  given by
\begin{equation}\label{eqn:Ya}
    dW_{t} = r\left( W_{t} - u(c_t) + h(a_t) + \theta^{A}_{t}Y_{t} \right)dt + r  \sigma Y_{t} dB^{a}_{t}
\end{equation}
Here
$\theta_{t}^{A}$
describes
the agent's worst-case scenario under the contract $(c,a)$
for the increments to output according to
$dX_{t} = (a_{t} + \theta^{A}_{t})dt + \sigma dB^{a}_{t}$.
In the contract $(c,a)$,
the agent's promised value $W_{t}$
grows at the interest rate,
it increases by the value of the agent's effort
$h(a_t)$, and it is reduced by the value of the payment $u(c_t)$.
As noted, the sensitivity $Y_{t}$
describes  the change in the agent's continuation value
due to output uncertainty.
As the principal commits to long-term contracts,
and infers the agent's worst-case scenario $\theta^{A}_{t}$,
she commits to provide the agent with the continuation value $\{W_{t}\}$.

Due to aversion to ambiguity, the principal and the agent can have different
perceptions about the evolution of output even if they have the same
drift-sets.
Letting $\theta_{t}^{P}$ denote the principal's worst-case scenario
from the principal's perspective the output process has the following distribution
$dX_{t} = (a_{t} + \theta^{P}_{t})dt + \sigma dB^{a}_{t}$.
Under the principal's worst-case scenario, the agent's promised value $W_t$ therefore evolves according to the following diffusion:
\begin{equation}\label{eqn:Yp}
    dW^{P}_t = r (W_t + h(a_t) - u(c_t) + \theta^{P}_{t}Y_{t}) dt + r \sigma Y_{t} dB^{a}_{t}
\end{equation}

As $dB^{a}_{t}$ is standard Brownian motion
under the measure $P^{a}$,
\eqref{eqn:Yp} implies that
the drift of the agent's continuation value
under the principal's worst-case distribution,
is $r (W_t + h(a_t) - u(c_t) + \theta^{P}_{t} Y_{t})$
and its quadratic variation is $r^{2} Y^{2} \sigma^{2}$.
By Ito's lemma,
the principal's robust contracting problem~\eqref{eqn:PrincipalsMinMax}
is represented recursively as an  ordinary differential equation
with the max-min criterion,
which is of the Hamilton-Jacobian-Bellman-Isaac type (HJBI for short)\footnote{See \cite{evans1984differential} for more detailed material
on the mathematical development of this class of problems.}
and it is given as follows:
\begin{equation}\label{eqn:HJBI}
\begin{split}
\hspace{-5mm} r F(W;\phi) = \max_{(c,a,Y)\in \Gamma}
\min_{\theta \in \Theta(a;\phi)}
 \Big\{ r ( a + \theta - c )
 + r & F'(W;\phi) \left(W - u(c) + h(a) + \theta Y \right) \\
&+ \frac{F''(W;\phi)}{2} r^{2} Y^{2} \sigma^{2} \Big\}
\end{split}
\end{equation}
\quad \quad for any $W > 0$
subject to the boundary conditions:
\[
F(0;\phi) = 0, \ \ \ \text{ and } \ \ \ F(W;\phi) \geq F_{0}(W).
\]

In this formulation, as in Sannikov,
since the agent can always choose to exert zero effort,
the representation \eqref{eqn:Ya} of the agent's continuation value
implies that the agent's participation holds with $W > 0$.
Here the constraint set $\Gamma$
includes non-negative wage payments $c \geq 0$,
non-negative action recommendation $a \geq 0$,
and non-negative sensitivity $Y \geq 0$
that satisfies incentive-compatibility for the action $a$
according to the agent's decision rule for his action choice as in
\textrm{Proposition}~\ref{prop:IC-kappa}.

In this representation, the first term of the objective function is the
expected flow profits and it increases by expected output and
decreases by payments to the agent. On the other hand, the second term
reflects the contribution to the expected continuation profits through
delaying the payments to the agent, and the last term
captures the effect due to exposing the agent to payoff uncertainty.
The principal's aversion to ambiguity is reflected by pessimistic expectation
formed by taking minimum over drift sets.

The function $F_{0}(W)$ describes the principal's guaranteed value after retiring the agent
with the continuation value $W$.
The principal always has  available the option to retire the agent with any
value $W \in \left[0,u(\infty) \right)$ where $u(\infty) = \lim_{c \rightarrow \infty } u(c)$.
Retiring the agent with value $u(c)$, the principal pays the agent constant wage $c$
and allows him to choose zero effort.  The principal's payoff
from retiring the agent is given by $F_{0}(u(c)) = -c$.
As there is no ambiguity after retirement,
the function $F_{0}$ does not depend on the strength of ambiguity
$\phi$.
As a state variable the agent's continuation value $W$ evolves during the employment, the boundary condition on retirement
ensures that the principal triggers retirement whenever the principal's payoff
from doing so exceeds the principal's payoff from continuing to employ
the agent.

The next result shows that
the recursive equation \eqref{eqn:HJBI} describes a solution to the
the principal's robust contracting problem~\eqref{eqn:PrincipalsMinMax}
in the sequential formulation.
It provides a {verification argument} using
dynamic programming approach and characterizes
the principal's contracting problem recursively using
the agent's continuation value as a state variable.
In particular, following a standard verification argument in the analysis of continuous-time principal agent problems and adapting it to drift-ambiguity,
it shows that
HJB \eqref{eqn:HJBI} characterizes a solution to the principal's sequential contract design problem.

\begin{thm}\label{thm:OptimalContract}
\normalfont
Assume
that the HJB equation \eqref{eqn:HJBI} has a unique smooth solution
$F \in C^{2}({R_{+}})$
and that
the stopping region takes the form
$\mathcal{S} = \{0\} \cup [W_{gp},\infty)$ for some
finite $W_{gp} < \infty$, then
the solution $(c(W),a(W),Y(W))$
for $W \in (0,W_{gp})$
characterizes an optimal contract with a positive
profit to the principal.
Such a contract is based on the agent's
continuation value as a state variable, which
starts at $W_{0}$ and evolves according to
\[
dW_{t} = r (W_{t} - u(c_{t}) + h(a_{t}) + \theta^{A}_{t} Y(a_{t})) dt + r \sigma Y(a_{t}) dB^{a}_{t}
\]
with  the payments $c_{t} = c(W_{t})$
and the action $a_{t} = a(W_{t})$
with termination time ${\tau} := \inf\{t \geq 0: {W}_{t} \in \mathcal{S}\}$.
\end{thm}


This theorem says that the principal's (sequential) problem~\eqref{eqn:PrincipalsMinMax}
can be alternatively specified in recursive form given in the associated HJB~\eqref{eqn:HJBI}
using the agent's continuation value as the summary of the payoff relevant history.
Moreover, the optimal controls $(a)$ and $(c)$ to the HJB describe the optimal long-term contract.%

To establish this result
we use \citet{possamai2020there}'s dynamic programming approach
to study principal-agent problem in continuous time
and adapt it drift ambiguity.
In the formulation of \citet{possamai2020there} the agent's action controls the drift of the output diffusion, like in the current paper,
which follows the general approach of
\citet{lin2020random}
allowing that the agent's action can control both the drift and the volatility of output diffusion.%
\footnote{
The formulation of \citet{lin2020random} generalize
 to random time-horizon the formulation of principal-agent problem
 in continuous-time with a finite time-horizon developed by
\citet{cvitanic2018dynamic}. These contributions provide
justification for reducing the principal's contracting problem
into a standard stochastic control problem as introduced by
\citet{sannikov2008continuous}.}

A version of Theorem \ref{thm:OptimalContract}
in a classical setting where there is no drift ambiguity
is given by \citet{sannikov2008continuous}.
In the analysis of the HJB equation \eqref{eqn:HJBI},
unlike Sannikov's approach we do not require an upper-boundary constraint
to determines retirement with a smooth-pasting condition
$F'(W_{gp}) = F_{0}'(W_{gp})$.
As illustrated by \citet{possamai2020there},
such a smooth-pasting constraint as in Sannikov's approach can entail loss of generality.
Instead, following \citet{possamai2020there}'s dynamic programming approach
and for simplicity, we model retirement with a finite random horizon for the
contracting relationship.
A possible generalization to infinite valued random horizon
where with a positive probability
the contracting relationship does not terminate with a retirement
is beyond the scope of this paper.%
\footnote{
 \citet{possamai2020there} discusses
 that such a generalization to infinite random horizon
 in a classical setting of \citet{sannikov2008continuous}
 is possible as the dynamic contracting problem reduces to a standard stochastic control problem
 for which the classical tools of stochastic analysis applies (see their Remark 3.8.).
 By analogous reasoning, we conjecture that the formulation of dynamic contracting
 with drift ambiguity can be generalized to infinite random-horizon.}
Arguably, in many economic applications, modeling with a finite random horizon does not entail
loss of generality as contracting relationships typically come to an end in a finite time.

%

We also assume that the HJB equation \eqref{eqn:HJBI}
has a smooth solution
$F \in C^{2}(R_{+})$. In the classical setting without drift ambiguity,
analogous HJB equation admits a smooth solution.
We conjecture that by adopting
Possama{\"\i} and Touzi's  approach using
classical tools of stochastic analysis, smooth solution obtains in our framework.
However, such a generalization is beyond the scope of the current paper

The analysis of HJB functional equation in~\eqref{eqn:HJBI}
yields further basic properties of the optimal long-run contract:
both the principal and the agent perceive the lower envelope as common worst-case scenario;
and the principal's value function $F$ displays monotonicity with respect to the strength of ambiguity.
We next take up these properties of the optimal contract
and then turn to perform monotone comparative statics analysis with respect
to the strength of ambiguity.


\subsection{Agreeing on the worst-case scenario}
\label{Sec:Agreeing on the worst-case scenario}
In the contracting relation, a priori the principal and the agent do not have to agree on
the worst-case scenario in drift sets due to the differences in objectives.
From the agent's perspective the worst-case minimizes the expectation of payments received
net of effort costs. We have already seen in \normalfont{Lemma}~\ref{lemma:Alower-envelope} that
non-trivial incentive provision implies that the agent's worst-case perception
is always the lower envelope of the drift set.
The principal's worst-case drift on the other hand minimizes the expectation of profits
and need not be the same as the agent's worst-case scenario.

The recursive representation~\eqref{eqn:HJBI} of the optimal contract provides a tractable
characterization for the worst-case scenario perceived by the principal.
For any incentive compatible action strategy $\{a_{t}\}$, the worst-case drift scenario perceived by the principal
$\theta^{P}$ is the drift that minimizes her expected payoff over the drift sets:
\[
\theta^{P}(a;\phi) = \argmin_{\theta \in [-\kappa(a;\phi),\kappa(a;\phi) ] } \theta \big\{1 + F'(W;\phi)Y \big\}.
\]
Equivalently,
\[
\theta^{P}(a;\phi) =
\left\{ \begin{array}{rl}
- \kappa(a;\phi) & \text{  if  } 1 + F'(W;\phi)Y \geq 0; \\
  \kappa(a;\phi) & \text{  if  } 1 + F'(W;\phi)Y < 0.
\end{array}
\right.
\label{eqn:L} \tag{L}
\]
Here the principal's worst-case analysis involves two terms of HJB\eqref{eqn:HJBI} due to drift ambiguity:
the first term involving $1$ reflects the contemporaneous effect on the expected output
and the other involving $F'(W;\phi)Y$ captures effect on the expected continuation profits.\footnote{
Notice also that there is no effect through the third and the last term of HJB \eqref{eqn:HJBI} since
there is no ambiguity about the volatility of the output diffusion (by assumption for tractability).}
The condition~\eqref{eqn:L} says that
the lower-envelope is the worst-case scenario for the principal, so long as
for each euro(/dollar) wage paid today the principal expects to make positive profits
from continuing to employ the agent.
Otherwise, additional payments incur net losses
and in this case the upper-envelope (which maximizes expected losses)
is her perceived worst-case scenario.
However, the latter does not arise \emph{during} employment in any optimal contract
since the principal always has the option of retiring the agent.

Early in the contracting relationship when the agent's continuation value $W$ is sufficiently low,
marginal profitability from delaying payments contingent on future performance
is positive $F'(W;\phi) > 0$ as it incentivizes the agent to work.
The principal's perceived worst-case scenario for low promised-values is therefore
the lower envelope as it minimizes the expectation of marginal profits.
Since the profit function is concave, for high enough $W$, the marginal benefit of delaying payments become negative $F'(W;\phi)<0$
as incentivizing the agent who has accumulated high continuation value
becomes costlier due to the income effects.
For high enough promised continuation value the principal's payment made today exceeds its
expected benefit in the continuation value, and hence her worst-case scenario is then
the upper envelope which maximizes such expected losses. However, at such high continuation values the
principal would rather retire the agent. Therefore, within the contracting relationship the principal's
worst-case perception is always the lower envelope of the set of drift terms.
The following result summarizes this observation:
\begin{prop}
\label{prop:common-worst-case}
\normalfont
In any optimal long-term contract, the principal and the agent both
perceive the lower envelopes of drift sets as the common worst-case scenario:
\[
{\theta^{A}(a;\phi) = \theta^{P}(a;\phi) = \min\big\{\theta \in \Theta(a;\phi) \big\} = a - \kappa(a;\phi)} \text{  for each  } a \in \mathcal{A}.
\]
\end{prop}
It is worth remarking that the result in \normalfont{Proposition}~\ref{prop:common-worst-case}
shows that the principal and the agent both perceive
the lower envelope to be the common worst-case scenario of drift sets
\emph{during} employment relationship.
As this agreement on the worst-case scenario is
an endogenous property in the optimal contract,
the model  is not equivalent to
a formulation that obtains by taking the lower-envelope as the specification of technology
in the model of \cite{sannikov2008continuous}.
Despite this difference,
however,
\normalfont{Proposition}~\ref{prop:common-worst-case}
implies that
the current model becomes observationally equivalent to the model in
\cite{sannikov2008continuous}.
In the current setting the principal can perceive the upper envelope
as her worst-case for high enough continuation values where
the cost of incentives exceed the benefits to the principal.
In such a case, the principal optimally chooses to retire the agent.
Therefore, the principal and the agent can use different worst-case scenarios
outside of employment.
As outside of employment there is no ambiguity,
this does not bring new possibilities into the principal's contract design.

\subsection*{On concavity of the principal's value function}
In a model of dynamic moral hazard where there is no drift ambiguity
\citet{sannikov2008continuous} and  \citet{possamai2020there}
show that the solution $F$ to the HJB equation is concave.
This property extends to the current model with drift-ambiguity.
Using the characterization of the principal's worst-case scenario in \eqref{eqn:L}
and rearranging \eqref{eqn:HJBI} implies
\[
\begin{array}{cl}
F''({W}) & =
\underbrace{\frac{F({W}) - a(W) + c(W) - F'({W})\left({W} - u(c(W)) + h(a(W)) \right)}{r \sigma^{2}Y(W)^{2}/2}}_{<0 } \\
&
\underbrace{- \kappa(a(W))
  {\frac{ | 1 + F'({W})Y(W) | }{r \sigma^{2}Y(W)^{2}/2}}   }_{< 0} < 0
\end{array}
\]
When there is no ambiguity, i.e., $\kappa = 0$,
in the expression above the second term
on the right-hand side disappears
and the model reduces to the classical model of \citet{sannikov2008continuous}
and hence the first term is negative.
In the presence of ambiguity, the value function is also concave
and in particular more concave relative to the classical case.
In \citet{miao2013robust}'s framework where the agent is ambiguity neural and the principal's
ambiguity sensitive preferences are modeled with variational preferences
based on a representation of \citet{maccheroni2006ambiguity},
the value function need not be concave.
In particular,
\citet{miao2013robust} observe
that
the value function can be convex whenever
the principal's ambiguity aversion is high enough.
This in turn implies that the principal can benefit from
using a public randomization device in the optimal contract.
In contrast, in the framework considered here
both the principal and the agent have MaxMin preferences
the principal's value function coincides with its concave hull
and hence such a public randomization cannot (strictly) improve upon her payoff.%
\footnote{
For the analogous argument in Bayesian formulation
see in particular \citet{sannikov2008continuous}'s Lemma 5 and
his remark following Proposition 4.}

\subsection{Basic monotonicity properties}
Using the common worst-case scenario from \normalfont{Proposition}~\ref{prop:common-worst-case}
basic properties on the monotonicity of the principal's value with respect to the strength of drift ambiguity flows from the analysis of HJB \eqref{eqn:HJBI}.
\begin{prop}\label{prop:Basic}
\normalfont
In the optimal long-term contract with drift ambiguity, the principal's value has the following
features with respect to a change in the strength of drift ambiguity:

\begin{enumerate}[label=(\normalfont{B}{\arabic*}),ref=(\normalfont{B}{\arabic*})]
\item \label{prop:kappamonotonicity}
Higher ambiguity reduces the principal's profits:
$\widehat\phi \succ_{\Amb} \phi$ implies $F(\cdot;\widehat \phi) < F(\cdot;\phi)$; \vspace*{-2mm}

\item \label{prop:kappaslopes} Higher ambiguity reduces the increments of the principal's profits:
$\widehat\phi \succ_{\Amb} \phi$ implies $F'(\cdot;\widehat \phi) < F'(\cdot;\phi)$; \vspace*{-2mm}

\item \label{prop:kappaconcavity} Higher ambiguity leads to more concave profit function:
$\widehat\phi \succ_{\Amb} \phi$ implies $F''(\cdot; \widehat \phi) < F''(\cdot; \phi)$. \vspace*{-2mm}
\end{enumerate}
\end{prop}

Here inequalities hold strictly except for the common initial value $W_{o}$.
While the first two properties that higher degree of ambiguity reduces the principal's guaranteed value
and its increments, in parts the
parts \ref{prop:kappamonotonicity} and \ref{prop:kappaslopes}
of Proposition~\ref{prop:Basic}
respectively,  are intuitively, the fact that increasing ambiguity uniformly makes the
principal's profit function more concave is probably less so.
Intuitively, increasing drift ambiguity for each action uniformly lowers the lower envelope which in turn reduces
expected profits monotonically, and also increases the agency cost of implementation by exposing the agent to larger uncertainty.
These two effects reinforce each other and imply increase in the curvature of the principal's value function. The parts
\ref{prop:kappaslopes}
and \ref{prop:kappaconcavity}
of Proposition~\ref{prop:Basic} are broadly technical
and used in the comparative static analysis.

If the mapping $\phi \mapsto \kappa(a;\phi)$ is differentiable,
the results in \textnormal{Proposition}~\ref{prop:Basic} can be expressed
in its differential form. In particular, for a differential change in $\phi$
that expands drift-set or increases $\kappa(a;\phi)$ for each action $a$,
the resulting changes in the principal's profit function corresponding
to the parts of the proposition are denoted as
$F_{\phi}(W;\phi) < 0$,
$F'_{\phi}(W;\phi) < 0$, and
$F''_{\phi}(W;\phi) < 0$, respectively.
The differential form will make it easier to develop
the comparative static analysis in the next section.

Before we turn to monotone comparative static analysis on
the optimal contract with respect to the strength of ambiguity, we give an additional
implication that continuation domain shrinks in the strength of ambiguity.

More specifically, define the \emph{continuation domain} at time $t$ by
\[
C(\phi) = \{W_t: F(W_t ;\phi) \geq F_0 (W_t) \}
\]
where $F_0(W_t)$ is the principal's outside option value from retiring the agent
that yields the agent utility $W_t$. Since the retirement option is always available to
the principal, $F_0(W_t)$ determines the lower bound on the principal's value.
The principal optimally retires the agent whenever the principal's value from the relationship $F(W_t;\phi)$ hits this lower boundary.
The set $C(\phi)$ therefore describes the states at which stopping and terminating the contract at time $t$ is strictly suboptimal.
The part~\ref{prop:kappamonotonicity} of \textrm{Proposition}~\ref{prop:Basic}  then implies that the continuation domain shrinks as the degree of ambiguity increases.
\begin{prop}
\label{thm:ContinuationDomain}
\normalfont
In any optimal long-term contract, the continuation domain shrinks in
the strength of ambiguity:
$C(\widehat\phi) \subset C(\phi)$ for all $t$ whenever $\widehat\phi \succ_{\Amb} \phi$.
\end{prop}

\section{High-Powered Incentives and Ambiguity}
\label{section:Monotonce}
This section presents the main results of the paper on the
nature of the optimal contract and
show that broadly speaking ambiguity flattens the power of incentives.
This result follows from extending and applying monotone comparative statics analysis
in continuous time framework inspired by \citet{quah2013discounting}'s approach.
In the context of single-agent decision problem, \citeauthor{quah2013discounting}'s powerful
analysis shows a monotonic relationship between optimal control variables and time-preference parameters.
For this development their key idea
is  that HJB equation describing the agent's
decision problem can  alternatively be viewed as an objective function
for a monotone comparative static analysis.
We start with adapting and extending this observation to
the agency problem here and then analyze monotonicity between optimal compensation mix and strength of ambiguity.

The principal's optimal contract problem using the lower-envelope
as the common worst-case scenario and
Ito's formula on HJBI~\eqref{eqn:HJBI}is rewritten as
\begin{equation*}
\begin{split}
\max_{(a,Y,c)\geq 0} \Big\{ r ( a - \kappa(a;\phi) - c)
+ & r F'(W;\phi)(W - u(c) + h(a) - \kappa(a;\phi)Y) \\
&+ r^{2}\sigma^{2} F''(W;\phi)Y^{2}/2 - r F(W;\phi) \Big\} =0
\end{split}
\end{equation*}
subject to the same boundary conditions and the action recommendation $a$ and
the required sensitivity $Y$ to implement it are related via the incentive compatibility~\eqref{eqn:IC-Y}.
The latter implies the agent's decision rule
for action choice $a(Y;\phi)$ as a function of incentives $Y$ and the degree
of ambiguity $\phi$.
Using this relationship
and letting $m(Y;\phi) = \left(W - u(c) + h(a(Y;\phi)) - \kappa(a(Y;\phi) ;\phi)Y \right)$
denote the drift of the agent's continuation value (to simplify notation)
we can express the contracting problem in terms of the
compensation mix in the following form:
\begin{equation}\label{eqn:H-Principtheta-Y}
\begin{split}
\max_{(Y,c)\geq 0} \Big\{ r(  a(Y;\phi) - \kappa(a(Y;\phi) ;\phi) - c)
 +  & r F'(W;\phi) m(Y;\phi)  \\
 + &r^{2} \sigma^{2} F''(W;\phi)Y^{2}/2 - r F(W;\phi) \Big\} =0
\end{split}
\end{equation}
%
The interest is in analyzing qualitative monotonic between the strength of ambiguity
and the optimal compensation mix. The model of drift ambiguity described so far
makes assumptions in \ref{assmp:kappa-A1} through \ref{assmp:kappa-A4}
about the strength of ambiguity or the length of drift sets $\kappa(a;\phi)$ up to second-order derivatives.
These assumptions enable tractable characterization of monotonic relationships between
the strength of ambiguity and the incentives to the agent for implementation given in \normalfont{Lemma}~\ref{lemma:Mon-Y} on the one hand,
and the principal's value function in \normalfont{Proposition}~\ref{prop:Basic} on the other hand.

Monotone comparative statics analysis in its full generality, however,
requires the knowledge of higher order derivatives beyond the second
due to the endogenous relationship \eqref{eqn:IC-Y} connecting actions to incentives.
Such a complication in principle is not difficult to solve and yet it does not yield
transparent economic analysis.
Instead, sharper characterization follows from focusing
on linear relationships between actions and
the degree of ambiguity. For this we assume a linear relationship
in how actions translate to strength of ambiguity:
$\kappa(a;\bphi) = \phi_{0} + \phi_{1}a$,
where the vector $\bphi := (\phi_{0},\phi_{1})$ describes the parameterization
of drift-ambiguity.
In this linear specification, the constant $\phi_{0} \geq 0$
denotes base-level drift-ambiguity independent of action,
and the slope $\phi_{1}$ describes how the agent's effort
affects the strength of drift-ambiguity.
Consistent with the assumptions~\ref{assmp:kappa-A1} through \ref{assmp:kappa-A4},
we consider $\bphi$ so that $\kappa(a;\bphi) \geq 0$ and
$a - \kappa(a;\bphi) \geq 0$, or equivalently, $\phi_{0} \geq 0$ and $\phi_{1} \leq 1$.
For the linear drift-ambiguity model, we have the following characterization
of monotone comparative statics.

\begin{thm}
\label{thm:Flattening-Linears}
\normalfont
Consider linear drift ambiguity model $\kappa(a;\bphi) = \phi_{0} + \phi_{1} a$ with $\phi_{0} \geq 0$.
In the optimal long-term contract,
as $\phi_{0}$ or $ \phi_{1} $ increases,
the optimal sensitivity $Y(W; \bphi)$ decreases,
for low enough continuation values such that $W \in [W_0,\upbar{W}]$.
On other hand, as it increases the optimal sensitivity $Y$ may be non-decreasing, for large enough continuation value such that
$W \in [\upbar{W},W_{gp}]$.
Here the principal's profits obtain the maximum at  the unique value
$\upbar{W} \in (W_0,W_{gp})$. \vspace{-2mm}
\end{thm}

This result says that the optimal sensitivity of the agent's value to output realizations broadly
becomes flatter as the strength of ambiguity increases except possibly
if the agent's continuation value is very high.
As lower-envelop of drift-sets becomes smaller, or the expected productivity
of each action in the worst-case scenarios decreases,
the power of incentives become
less responsive or less contingent to output realizations.

There is an intuitive explanation for why ambiguity flattens power of incentives.
Drift ambiguity has an contemporaneous effect and ambiguity aversion reduces the expected benefits from
each action due to pessimistic expectation. Drift ambiguity also has a continuation effect and
it reduces the principal's value because the principal has to provide
stronger incentives $Y$ to implement any viable action according to
\normalfont{Lemma}~\ref{lemma:Mon-Y}.
The need for strong incentives contribute to the agency costs of implementing actions by exposing
the agent to larger output uncertainty captured by the term $F''(W;\bphi)Y^{2}$.
All else being equal, both the contemporaneous and the continuation effects
reduce the marginal profitability and by the monotone comparative statics argument the
principal optimally chooses to provide lower powered incentives under higher drift ambiguity.
However, not all else remain equal as  there is an additional continuation
effect of ambiguity interacting  with marginal profitability of delayed payments
captured by the term involving $F'(W;\bphi)$.
This effect due to delayed payments depends on the level of promised value $W$.

For low enough continuation values, at the beginning of the relationship,
delaying payments are beneficial for the principal as $F'(W;\phi) > 0$.
Since higher ambiguity makes the back-loaded payments less desirable for the agent
due to ambiguity aversion  (see \normalfont{Proposition}~\ref{prop:Wdiff-Ambig-A},
and discussion following it),
this delayed-payment effect reinforces the two effects previously identified,
and further reduces incentive benefits of the continuation value; the
principal optimally chooses lower powered incentives.
If, on the other hand, the continuation value is sufficiently high that
the delaying payments become costly for the principal due to wealth effects, i.e., $F'(W;\bphi) < 0$,
increases in ambiguity has a positive effect on the principal's profits.
Therefore, depending on the relative strength of this effect from delaying payments
in comparison to the contemporaneous and the continuation effects,
the strength of incentives may increase in ambiguity for high enough continuation values.
However, in typical employment relationships analyzed by \citet{sannikov2008continuous},
where the optimal incentives are first increasing and then decreasing
in continuation value, the optimal incentives become flatter with increases in ambiguity.

%
%


Drift ambiguity has a similar flattening effect on the evolution of optimal effort stream
and reduces the optimal effort profile.
When drift ambiguity increases in action in
the linear model $\kappa(a;\bphi) = \phi_{1} a$, from the incentive compatibility
condition~\eqref{eqn:IC-Y} the optimal action is given by $a = (1-\phi_{1})Y$
and higher ambiguity (corresponding to higher $\phi$) implies  decrease in optimal action.
If, on the other hand, the strength of ambiguity decreases in actions,
that is the size of the drift set is of the form
$\kappa(a;\bphi) = \phi_{0} - \phi_{1} a$,
then  the incentive compatibility condition~\eqref{eqn:IC-Y} implies
$a = (1+\phi_{1})Y$, which again decreases in the strength of
ambiguity as in this case smaller $\phi_{1}$ corresponds to
larger drift sets.

Drift ambiguity has an ambiguous effect on the duration of the agency relationship, or the optimal stopping time.
Notice that \normalfont{Theorem}~\ref{thm:Flattening-Linears}
says that expected worst-case increment in the agent's continuation value (drift of the diffusion $W$)
and its volatility $Y$ both becomes smaller.
While the first effect reduces the length of time it takes
to trigger the end of the relationship via retiring or firing the agent,
the second effect delays such an event.
The effect of ambiguity on the duration of the contracting relationship is therefore ambiguous,
even if according to \normalfont{Proposition}~\ref{thm:ContinuationDomain}
the upper bound of the continuation domain is smaller with ambiguity.

The analysis so far has characterized the effect of ambiguity on optimal power of incentives
and on the duration of the contracting relationship.
Next we further characterize the optimal payment schemes and show that the presence of ambiguity
depresses wages and makes the optimal contract more reliant on the long-term incentives.
\begin{prop}\label{prop:Wage-Depressing}
\normalfont
Higher ambiguity depresses wages:
$c(W;\widehat\bphi ) < c(W;\bphi)$  for all $W$ whenever $\widehat\bphi \succ_{\Amb} \bphi$.
Moreover, the ratio of volatilities of the agent's consumption and continuation values is greater
in the contract associated with $\widehat\bphi$
whenever $c(\widehat{W};\widehat\bphi) = c(W;\bphi)$.
\end{prop}

This result is reminiscent of \citet{sannikov2008continuous}'s Theorem 4
that studies how the contractual environment affects the optimal long-term incentives.
In particular, imprecise information about technology
as compared to precise information
compresses wage profile, and
reduces the sensitivity of current wage with respect to the continuation value.
The latter implies that the contract associated with precise information environment
relies less on short-term incentives and more on long-term incentives.

Turning to the interpretation of the result in Proposition~\ref{prop:Wage-Depressing} notice that
as in the classical case, the variation in the continuation value with output realizations is still
 present to incentivize the agent's effort, see from the incentive compatibility condition
 \eqref{eqn:IC-kappa-A}. The familiar intuition is that to solve agency conflict the principal's
 design of contract aligns interests by letting agent's compensation to positively vary
with her profits. This intuition is still applicable to interpret our result
as during employment the parties endogenously agree on the worst-case.
On the other hand, as we have seen the presence of ambiguity and aversion to it
makes back-loaded incentives a less effective tool relative to the classical case
because ambiguity-averse agent prefers certainty of payments today
over uncertainty of future payments.

\section{Extension to Heterogeneous Drift-Ambiguity}
\label{Section:Heterogeneous Drift sets}
The agency problem posed with symmetric ambiguity and the approach to its analysis
is flexible enough to allow for heterogeneity in  ambiguity.
In particular, we now allow for drift-sets perceived by the principal and the
agent to be different by setting for each action $a \in \actions$,
$\kappa(a;\phi^{A})$
and
$\kappa(a;\phi^{P})$, respectively.
We assume that the parties share common knowledge about the heterogeneity
in ambiguity.
The principal and the agent can evaluate a contract
using different drift sets.
To keep track of drift-ambiguity perceived by the parties,
we use the notation $\bphi := (\phi^{A},\phi^{P})$
by extending the notation used in the symmetric case.

For tractability, we make regularity assumptions analogous to
those in the symmetric case and
assume that Assumptions \ref{assmp:kappa-A1}-\ref{assmp:kappa-A4} hold for
the mappings $\kappa(a;\phi^{A})$ and $\kappa(a;\phi^{P})$, respectively.

With this reformulation Proposition~\ref{prop:Wdiff-Ambig-A},
which uses only the drift set of the agent, applies using the drift set $\Theta(a_t;\phi^{A}) = [a_t - \kappa(a_t;\phi^{A}), a_t  + \kappa(a_t;\phi^{A})]$
and yields a representation of the agent's value as a diffusion.
Using this diffusion representation then,
Lemma~\ref{lemma:Alower-envelope} applies analogously and
the agent's worst-case scenario for any non-zero contract
is now given by the lower-envelope of his drift-set:
$\theta^{A}(a_t;\phi^{A}) = a_t - \kappa(a_t;\phi^{A})$.
Applying \normalfont{Proposition}~\ref{prop:IC-kappa}
incentive compatible sensitivity $Y$ to implement a non-trivial action $a >0$
analogous to \eqref{eqn:IC-Y} in with minimized cost
is therefore characterized by
\begin{equation}
\label{eqn:IC-Y-Asy}
Y = \frac{a}{1 - \kappa_{a}(a;\phi^{A})}.
\end{equation}
This in turn extends Lemma~\ref{lemma:Mon-Y}
and implies that as the ambiguity perceived by the agent increases
the required variation $Y(a;\phi^{A})$ to implement an action $a$, which solves for $Y$ the equation \eqref{eqn:IC-Y-Asy}, increases
and that for any $\phi^{A}$ the  incentives $Y(a;\phi^{A})$ increase in $a$.
Notice that for any contract the agent's decision rule given
by \eqref{eqn:IC-Y-Asy} depends only on the agent's ambiguity
$\phi^{A}$ but not on the principal's ambiguity $\phi^{P}$.

Using the characterization of the agent's decision rule for his action choice,
HJBI formulation for the optimal contracting problem
for heterogeneous ambiguity $\bphi = (\phi^{A},\phi^{P})$ then
takes the following form
\begin{equation}\label{eqn:HJBI-Asy}
\small{
\begin{split}
r F(W;\bphi) = \max_{(a,c)\geq 0 } \min_{\theta \in \Theta(a;\phi^{P})}
  \Big\{ r ( a + \theta - c ) + & r F'(W;\bphi) \left(W - u(c) + h(a) + \theta Y(a;\phi^{A}) \right) \\
&+ \frac{F''(W;\bphi)}{2} r^{2} Y(a;\phi^{A})^{2} \sigma^{2} \Big\}
\end{split}
}
\end{equation}
subject to the boundary conditions and the smooth pasting condition:
\[
F(0;\bphi) = 0, \ F(W_{gp};\bphi) = F_{0}(W_{gp}) \ \text{and} \ F'(W_{gp};\bphi) = F'_{0}(W_{gp}).
\]
Comparing the contracting problem~\eqref{eqn:HJBI} with symmetric ambiguity
notice here the distinct roles played by $\phi^{A}$
and $\phi^{P}$.
The principal's value depends on her ambiguity $\phi^{P}$ directly
through his worst-case expectation over $\Theta(a;\phi^{P})$.
Her value indirectly depends on the agent's ambiguity
$\phi^{A}$ as it determines via \eqref{eqn:IC-Y-Asy}
the agent's incentive compatible
decision rule.
The principal's optimal contract design takes into account both sources of ambiguity.

The argument for
\normalfont{Theorem}~\ref{thm:OptimalContract}
extends and applies to heterogeneous case
as the regularity conditions analogously hold.
The recursive HJBI formulation~\eqref{eqn:HJBI-Asy} therefore
provides an equivalent characterization of the optimal contracting problem.
In turn,
the recursive representation~\eqref{eqn:HJBI-Asy}
provides a tractable
characterization for the worst-case scenario perceived by the principal
as in the symmetric case.
For any incentive compatible action strategy $\{a_{t}\}$
and the corresponding incentives $(Y_t)$,
the worst-case drift scenario $\theta^{P}$ perceived by the principal
is the drift term that minimizes her expected payoff over the drift set:
\[
\theta^{P}(a;\bphi) = \argmin_{\theta \in \left[-\kappa(a;\phi^{P}),\kappa(a;\phi^{P}) \right] } \theta \big[1 + F'(W;\bphi)Y \big]
\]
Equivalently,
\[
\theta^{P}(a;\bphi) =
\left\{ \begin{array}{rl}
- \kappa(a;\phi^{P}) & \text{  if  } 1 + F'(W;\bphi)Y \geq 0; \\
  \kappa(a;\phi^{P}) & \text{  if  } 1 + F'(W;\bphi)Y < 0.
\end{array}
\right.
\label{eqn:L'} \tag{L$'$}
\]
The interpretation of this condition is analogous to that, \eqref{eqn:L},
in the symmetric ambiguity case:
\emph{during} employment in any optimal contract
the lower-envelope of drift-set is the worst-case scenario for the principal.
For otherwise, additional payments incur net losses
and the principal always has the option of retiring the agent instead of making such losses.

\subsection{Basic properties with heterogeneous ambiguity}

Proposition~\ref{prop:Basic} has established basic monotonicity properties
of the principal's value function with respect to symmetric drift ambiguity.
Its analysis extends to heterogeneous ambiguous perceptions
and shows that qualitatively the analogous results hold.

To represent different degrees of drift ambiguity in heterogeneous case,
we define the relation $\widehat\bphi \succ_{\Amb} \bphi$ to mean
contracting environment
where at least one of the parties perceive larger drift ambiguity:
formally,
$\kappa(\cdot ;\widehat{\phi}^{i}) \geq \kappa(\cdot;\phi^{i})$
for each $i = A ,P$ and
$\kappa(\cdot;\widehat\phi^{i}) > \kappa(\cdot;\phi^{i})$ and
for at least one of $i = A ,P$.
This naturally extends to the heterogeneity in drift ambiguity
the definition of comparative ambiguity
for the symmetric case.

Using the lower envelopes as worst-case scenarios for each party,
analogous to \textnormal{Proposition}~\ref{prop:Basic}
the basic monotonicity properties for the principal's guaranteed value
with respect to the strength of ambiguity $\bphi = (\phi^{A},\phi^{P})$
follow from the analysis of HJBI \eqref{eqn:HJBI-Asy}:

\begin{prop}\label{prop:Basic-Asy}
\normalfont
In the optimal long-term contract with drift ambiguity, the principal's guaranteed value
has the following features, \vspace{-2mm}
\begin{enumerate}[label=(\normalfont{B}{\arabic*}$'$),ref=(\normalfont{B}{\arabic*}$'$)]
\item \label{prop:kappamonotonicity-Asy}
Higher ambiguity reduces the principal's profits:
$\widehat\bphi \succ_{\Amb} \bphi$ implies $F(\cdot;\widehat \bphi) < F(\cdot;\bphi)$; \vspace*{-2mm}

\item \label{prop:kappaslopes-Asy} Higher ambiguity reduces the increments of the principal's profits:
$\widehat\bphi \succ_{\Amb} \bphi$ implies $F'(\cdot;\widehat \bphi) < F'(\cdot;\bphi)$; \vspace*{-2mm}

\item \label{prop:kappaconcavity-Asy} Higher ambiguity leads to a more concave profit function:
$\widehat\bphi \succ_{\Amb} \bphi$ implies $F''(\cdot; \widehat \bphi) < F''(\cdot; \bphi)$. \vspace*{-2mm}
\end{enumerate}
\end{prop}
Here inequalities hold strictly except for the common initial value $W_{o}$.
In other words, higher ambiguity perceived by the principal or the agent reduces the principal's value function, decreases its increments, and makes it more concave.
These basic properties with heterogeneous drift ambiguity are qualitative
similar to those with homogeneous drift ambiguity given in \textnormal{Proposition}~\ref{prop:Basic}.
Relative to the latter allowing for heterogeneous drift-ambiguity in the analysis
of HJBI~\eqref{eqn:HJBI-Asy}
highlights two distinct effects of the parties' ambiguity perceptions.
Firstly, increases in the agent's ambiguous perceptions
increase agency costs due to higher powered incentives required for implementation
and hence reduce the principal's value.
Secondly, increases in the principal's ambiguity reduces her guaranteed value
by deteriorating expectations in her worst-case scenario.


\subsection{Optimal incentives and heterogeneous  drift-ambiguity}
\label{subsection:Monotonce-Het}

We now turn to analyze the monotone comparative static
analysis on the strength of ambiguity and the power of incentives in
the optimal compensation mix.
For this,
using the basic monotonicity properties
of HJBI's objective function from~\normalfont{Proposition}~\ref{prop:Basic-Asy},
we extend Theorem~\ref{thm:Flattening-Linears} to heterogenous drift ambiguity.
This analysis shows broadly that the qualitative feature of flattening of incentives
at higher ambiguity levels holds allowing for heterogeneity in drift-ambiguity.

\begin{thm}
\label{thm:Flattening-Linears-Het}
\normalfont
Consider a linear drift ambiguity model
$\kappa(a;\bphi^{i}) = \phi^{i}_{0} + \phi^{i}_{1} a$
with $\phi^{i}_{0}  \geq 0$
for each $i = A,P$
In the optimal long-term contract, the incentives with respect to ambiguity have the following characterization. \vspace{-2mm}
\begin{enumerate}
\item[(a)] As $\phi_{0}^{P}$, $\phi^P_{1}$ or  $\phi^A_{1}$ increases,
the optimal sensitivity $Y(W; \bphi)$ decreases
for low enough continuation values $W \in [W_0,\upbar{W}]$.
On other hand, as it increases the optimal sensitivity may be non-decreasing, for large enough continuation value such that.
The principal's profits obtain the maximum at  the unique $\upbar{W} \in (W_0,W_{gp})$. 		\vspace{-2mm}
\item[(b)] As $\phi_{0}^{A}$ increases
the optimal sensitivity $Y(W; \bphi)$ decreases
for all continuation values during employment $W \in [W_0,\upbar{W}]$.           \vspace{-2mm}
 \end{enumerate}
\end{thm}

This characterization
is qualitatively similar to that of the symmetric linear-drift-ambiguity model described in
Theorem~\ref{thm:Flattening-Linears}.
Formulation with heterogeneous
ambiguity perceptions reveals   different possible effects
of the strength of ambiguity perceived by the parties.
Compared to the symmetric formulation,
notice that the agent's heterogeneous constant ambiguity perception
has unambiguous effect and flattens the incentives.
This is because an increase in the agent's constant ambiguity $\phi_{0}^{A}$
worsens the expected value for the agent from any viable contract
and hence
constrains the set of contracts the principal can implement.

On the other hand, an increase in the constant ambiguity $\phi_{0}^{P}$ perceived by the principal
has an effect of incentives that can depend on the level of the promised value to the agent $W$.
First, increase in  $\phi_{0}^{P}$ lowers the expectation in the flow
benefit, which unambiguously call for lower incentives because
its benefits are now lower.
Second, the principal's worsening of lower
envelope also limits the increments to the agent's continuation
value from the principal's perspective.
This second effect's sign depends on where the relationship
stands during the contract.
For low enough $W \in [W_0,\upbar{W}]$
backloading payments is beneficial for the principal
and
an increase in  $\phi_{0}^{P}$ limits this gain by
reducing the agent's continuation value from the principal's perspective.
This reinforces the first effect and lowers the optimal incentive
provision.
On the other hand, for high enough $W$
backloading incentives is costly due to wealth effects
and in this case reducing the increments to the continuation
value as $\phi_{0}^{P}$ increases can decrease the cost of incentive provision.
This in turn counteracts the first effect
and hence can induce an increase  in incentives.
This effect is also present for changes in $\phi^P_{1}$ or  $\phi^A_{1}$
that increases drift-sets.
In either scenario,
incentive path becomes compressed, i.e., it takes values in a narrower range,
as constant ambiguity increases
since typically incentives are first increasing and then decreasing
in $W$ in a continuous manner (see Sannikov (2008)).

Applying similar arguments extends
\normalfont{Proposition}~\ref{prop:Wage-Depressing}
to the heterogenous drift-ambiguity case. Therefore, the optimal wage scheme
becomes more depressed as the strength of drift ambiguity perceived by the
principal and/or the agent increases.

\section{Concluding Remarks}

This paper presents a dynamic model of
moral hazard relationship whose central feature is
imprecise information about the consequences of actions.
Imprecise
information about technology translating actions into outcomes
are important for applications, yet known to lead to difficulties of analysis in
the classic Bayesian formulation of dynamic moral hazard problems.
Methodologically, the paper provides a model of dynamic
moral hazard that formulates imprecise information using drift-ambiguity
i) can sustain nontrivial dynamic incentives; ii) has a tractable solution, characterized by a one-dimensional differential equation; and iii)
incorporates novel monotone comparative statics in a dynamic agency relationship.

Substantively, the paper sheds new lights on the effects of imprecise information
and ambiguity aversion on dynamic contracts (compressing of the incentive effects of ``high-powered'' incentive payments), and we provide empirical predictions on the ``high-powered'' incentives (compressing  of the incentive payments in imprecisely understood environments
relative to Bayesian environments with precise understanding of connection between actions and their consequences).
This resonates well with empirical observations
that
in workplace settings where output and/or quality are difficult to measure,
for instance performance of white collared workers,
the workers' contracts are not sensitively fine-tuned to performance realizations
\citep{lazear2018compensation}.

We want to make a remark about our use of the MaxMin criteria as a formalization of attitude towards imprecise information or non-probabilistic beliefs. One class of axiomatizations of MaxMin preferences come from \citet{gilboa1989maxmin}'s approach towards ambiguity  that takes the state space as a primitive. We do our analysis on the space of payoff-relevant
outcomes, or more specifically over paths of output diffusions. In general, the state-space approach towards ambiguity is not conceptually amenable to agency problems, as the probabilities of outcomes in the latter are conditional on actions (see \citet{karni2006subjective} and \citet{karni2009reformulation} for discussions on this issue). We therefore do not invoke such axiomatizations as a foundation for our formulation, instead we see this simply as a plausible way to capture the idea of robustness concerns when designing contracts. Our take is that such concerns are more appropriately modeled as evaluations of outcome prospects.
We therefore follow the behaviorally justifiable approach of  taking as a primitive the decision maker's preferences over sets of distributions over the outcome space (see  \cite{ahn2008ambiguity}, \cite{olszewski2007preferences} and \cite{dumav2018ambiguity} for similar approaches).

For tractability the current model uses, for
ambiguity sensitive preferences
an extreme form of pessimism represented by the formulation with MaxMin Expected Utility.
In the dynamic framework considered here,  MEU is the only form that allows for non-separability between time and ambiguity attitudes \citep{strzalecki2013temporal}.
This provides therefore a transparent insight into effect of ambiguity on dynamic incentives.
However, the framework is analytically flexible enough to accommodate
other formulations of ambiguity sensitive preferences.

While interpreting our results, one has to be conscious of our formalization
 that the agent and the principal
have stationary perceptions of drift ambiguity.
Generally, our intuitions about the dynamics of the agent's wages and incentives
do not appear to depend on this assumption. However, the assumption that the parties have
stationary ambiguity matters if one addresses the practically important issues that the
contracting parties experiment about the technological possibilities.
These highly interesting issues introduce forms of non-stationarity in
the contracting problem and further analysis is left to future research.

\bigskip

\section*{APPENDICES}
\section*{The structure of the sets of priors: Rectangularity and regularity}\label{sect:set_of_priors}
\label{Appx:Drift-Sets}

In this section, we provide a class of multi-priors
that satisfy dynamic consistency and regularity conditions.
These conditions in turn ensure that the contracting
problem admits a solution and can be analyzed recursively.

The formalization of the drift-sets
follows the approach in \citet{chen_epstein_ambiguity2002}
and adapts it to the contracting problem.
Throughout we use
as primitive
a measure space $(\Omega, \mathcal{F}, P)$
and a Brownian motion $B$ defined on this space.

We begin with observing that as the contraction problem has a finite random horizon
$P(\tau < \infty) = 1$,
we fix any given finite time horizon $T \leq \infty$ such that $\tau(\omega) \leq T$ with probability one,
and let time $t$ vary over $[0,T]$.
We also fix any given drift ambiguity scenario $\phi$
and hence leave its dependence implicit to simplify notation.

For any action profile $a = \{a_{t}\}$ in an admissible contract,
we consider a drift process $\theta=\{\theta_{t}\}$
adapted to the filtration $\mathcal{F}$ such that
$\theta_{t} \in \Theta(a_{t})$.
Using the drift process
we define a process of Girsanov exponents as follows
\begin{equation}
\label{eqn:girsanov-exponent}
z_{t}^{a + \theta} :=
\exp\left\{
-\frac{1}{2} \int_{0}^{t} |a_{s} + \theta_{s}|^{2}ds - \int_{0}^{t} (a_{s} + \theta_{s}) dB^{a}_{s}
\right\},
\quad
0 \leq t \leq T,
\end{equation}
where $\theta_{t} \in \Theta(a_{t}) = [-\kappa(a_{t}), \kappa(a_{t})]$.

As the action set $\mathcal{A}$ is a compact subset in $R_{+}$
and the drift set $\Theta(a)$ for any action $a \in \mathcal{A}$
is a compact subset in $R_{+}$, the process $a + \theta : = \{a_{t} + \theta_{t}\}$ is bounded
and hence satisfies the following square-integrability condition
\[
\expect
\left[
\exp\left(
\frac{1}{2} \int_{0}^{T} |a_{s} + \theta_{s}|^{2}ds
\right)
\right] < \infty.
\]
Using this property, 
and using Girsanov's change of measure theorem
ensure that  the process $\{ z_{t}^{a + \theta} \}$ defines a $P-$density
on $\mathcal{F}_{T}$ (see, for instance, \citet[p. 337]{duffie2001dynamic}).
Therefore, the process $a + \theta$ generates a probability measure
$P^{a+\theta}$ on the measurable space of sample paths $(\Omega, \mathcal{F})$
and the Radon-Nikodym derivative of this measure
with respect to the reference measure $P$ is given by
\begin{equation}
\label{eqn:Radon-Nik-P}
\frac{d P^{a+\theta}}{d P}
\bigg\rvert_{\mathcal{F}_{t}} = z_{t}^{a + \theta},
\quad
0 \leq t \leq T.
\end{equation}

For the action $a = \{a_{t}\}$  and the measurable drift process
$\{\theta_{t}\}$ such that $\theta_{t} \in \Theta(a_{t})$
we have therefore constructed a measure $P^{a+\theta}$  equivalent to $P$
via a change of measure. Now, as in \citet{chen_epstein_ambiguity2002}
we define the set of multi-priors for the given action $a = \{a_{t}\}$
as the set of all such equivalent measures $P^{a+\theta}$
that can be constructed using some measurable selection $\{\theta_{t}\}$
from the process of correspondences $\Theta(a_{t})$.
The multi-prior for the action is then given by:
\begin{equation}
\begin{split}
\label{eqn:multiprior-a}
\mathcal{P}^{\Theta(a)} :=
\{
P^{a+\theta}: \{ \theta_{t} \}
\text{ is a measurable selection from }
\{
\Theta(a_{t})
\}
\text{ and } \\
P^{a+\theta}
\text{ is defined by }
\eqref{eqn:Radon-Nik-P}
\}
\end{split}
\end{equation}
For any admissible action $a = \{a_{t}\}$, by analogous argument  and terminology
used by \citet{chen_epstein_ambiguity2002}
the set of priors $\mathcal{P}^{\Theta(a)}$
satisfies rectangularity.

We next turn to verify that
the drift sets $\{\Theta(a_{t})\}$ and the set of priors
$\mathcal{P}^{\Theta(a)}$ satisfy regularity conditions
that facilitate the recursive analysis of the contracting problem.

As the action $a = \{a_{t}\}$ is measurable with respect to the filtration $\mathcal{F}$,
and the function $\kappa$ defined on $\mathcal{A}$ is continuous,
the process of compact and convex-valued correspondences
$\{\Theta(a_{t})\}$ is measurable with respect to the same filtration.%
\footnote{
Indeed, using the definition of measurability
in the class $\mathcal{K}$ of compact subsets of $R$
(see for instance, \citet[Chapter 14.12]{aliprantis2006infinite})
for each compact set $K \in \mathcal{K}$,
it is the case that
$\{(t,\omega) \in
[0,s] \times \Omega:
\Theta(a_{t}(\omega)) \bigcap K \not= \emptyset \} \in \mathcal{B}([0,s]) \times \mathcal{F}_{s}$.
}
Moreover, a normalization property
$0 \in \Theta(a_{t})$ holds almost surely as $0 \in \Theta(a) = [-\kappa(a), \kappa(a)] $
for any $a \in \mathcal{A}$. We summarize these regularity properties
on the drift sets as follows
\begin{lemma}
\label{lemma:basics-drifsets}
\normalfont
The process of drift sets $\{\Theta(a_{t})\}$
satisfies:

\begin{enumerate}
\item[(a)] $\{\Theta(a_{t})\}$ is uniformly bounded.

\item[(b)] The correspondence $(t,\omega) \mapsto \Theta(a_{t}(\omega))$ is compact-valued and convex-valued, and is $\mathcal{F}_{t}-$measurable for any $0 < t \leq T$.

\item[(c)] Normalization $0 \in \Theta(a_{t}(\omega))$ holds $dt \otimes dP$ a.e.

\end{enumerate}
\end{lemma}

We next turn to verify regularity properties for the set of priors $\mathcal{P}^{\Theta(a)}$.
For this we use the \emph{weak topology} on the space of countably additive probabilities
$ca(\Omega,\mathcal{F}_{T})$. To characterize this topology
we use the Prokhorov metric $\rho$.
For any two measures $Q$ and $Q'$ in the space $ca(\Omega,\mathcal{F}_{T})$
the Prokhorov distance $\rho(Q,Q')$ is defined by
\begin{equation} \label{eqn:prokhorov}
\begin{split}
\hspace{-5mm} \rho(Q,Q') :=
\inf\{
\varepsilon > 0:
\text{ for any closed set } F \in \mathcal{F}_{T}, \\
 Q(F) \leq Q'(F^{\varepsilon}) + \varepsilon
\text{  and  }
& Q'(F) \leq Q(F^{\varepsilon}) + \varepsilon
\}
\end{split}
\end{equation}
where $F^{\varepsilon}$ is the ${\varepsilon}-$ball around the set $F$.
The Prokhorov metric $\rho$ metrizes the weak topology on $ca(\Omega,\mathcal{F}_{T})$,
indeed $\rho(Q_{n},Q) \rightarrow 0$ if and only if
 for all bounded and continuous function $f$
 $\int_{\Omega} f dQ_{n} \rightarrow \int_{\Omega} f dQ$
 (see, instance, \citet[Theorem 9.3.2]{corbae2009introduction}).
Moreover, on the class compact and convex subsets of
$ca(\Omega,\mathcal{F}_{T})$ we use Hausdoff topology
corresponding to the  Prokhorov metric $\rho$.

Under the regularity conditions on the process of drift sets $\{\Theta(a_{t})\}$
in {Lemma~\ref{lemma:basics-drifsets}},
for any admissible action $a = \{a_{t}\}$, the following set of properties
 on the prior sets  $\mathcal{P}^{\Theta(a)}$
 hold
\begin{lemma}
\label{lemma:basics-multiprior}
\normalfont
The set of priors $\mathcal{P}^{\Theta(a)}$ satisfies:

\begin{enumerate}
\item[(a)] $P \in \mathcal{P}^{\Theta(a)} $ and $\mathcal{P}^{\Theta(a)}$ is uniformly absolutely
continuous with respect to $P$. Each measure in
$\mathcal{P}^{\Theta(a)}$ is equivalent to $P$.

\item[(b)] $\mathcal{P}^{\Theta(a)}$ is convex and compact in the weak topology.


\item[(c)] For every $\xi \in L^{2}(\omega, \mathcal{F}_{T},P)$,
there exists
$P^{a + \theta^{\ast}} \in \mathcal{P}^{\Theta(a)} $ such that
\begin{equation*}
\expect^{a + \theta^{\ast}}
\left[ \xi | \mathcal{F}_{t} \right] =
\min_{\theta \in \Theta(a) } \expect^{a + \theta}\left[ \xi | \mathcal{F}_{t} \right],
\quad
0 \leq t \leq T.
\end{equation*}
\end{enumerate}

\end{lemma}

\begin{proof}
The properties in (a), (b) and (c)
follows from the arguments analogous to those made in the proof of
\normalfont{Theorem~2.1} by \citet{chen_epstein_ambiguity2002}.
\end{proof}

\section*{Proofs of the Results for Symmetric Drift Ambiguity in Sections~\ref{section:contracting_problem}-\ref{section:Monotonce} }

\subsection*{Proof of {Proposition}~\ref{prop:Wdiff-Ambig-A}}

As in the structure of the prior sets,
we begin with observing that as the contraction problem has a finite random horizon
$P(\tau < \infty) = 1$, we fix any given finite time horizon $T \leq \infty$ such that $\tau(\omega) \leq T$ with probability one, and let time $t$ vary over $[0,T]$.
This is without loss of generality.
Indeed, notice that if $t \geq \tau(\omega)$, then
the agent stops working $a_{t}(\omega) = 0$ ,
and receives no payment $c_{t}(\omega) = 0$ or a positive payment $c_{t}(\omega) = \xi > 0$,
depending on whether at the stopping time $\tau(\omega)$
the principal terminates the contract by firing the agent or by retiring him, respectively.

For a given pair of processes of wages and actions $\{c_{t}\}$ and $\{a_{t}\}$, respectively,
define the agent's valuation process $\{V_t\}$ by
\begin{equation}\label{g-valuation_A}
V_{t} = r \int_{0}^{t} e^{-rs} ( u(c_{s}) - h(a_{s}) ) ds + e^{-rt}W_{t}(c,a)
\end{equation}
where $W_{t}(c,a)$ is the continuation value defined by
\begin{equation*}
W_{t} = \min_{Q \in \mathcal{P}^{\Theta(a)}} \expect_{Q}
\left[
\int_{t}^{\infty} e^{-r(s-t)}\left( u(c_{s}) - h(a_{s}) \right)ds \given \mathcal{F}_{t}
\right].
\end{equation*}
Since, by construction, $\mathcal{P}^{\Theta(a)}$ is a rectangular set of multi-priors over
the paths of output realizations, using \cite{chen_epstein_ambiguity2002}'s representation theorem,
then we find a progressively measurable process $\{Y_{t}\}$ such that
\begin{equation}\label{Agent's g-mgale}
-dV_{t} = - \theta^{\ast}_{t} r e^{-rt}|Y_{t}|dt - \sigma r e^{-rt} Y_{t} dB^{a}_{t}
\end{equation}
where $B_{t}$ is a Brownian motion under the reference measure $P$;
$\theta^{\ast}_{t}  := \min_{\theta_{t} \in \Theta(a_{t})} \theta_{t}|Y_{t}|$ is the
agent's worst-case drift for the action process $\{a_{t}\}$; and the factor
$r e^{-rt} \sigma$ is a convenient rescaling factor.
On the other hand, differentiating \eqref{g-valuation_A} with respect to $t$ one finds that
\begin{equation}\label{Agent's dV-path}
dV_{t} = r e^{-rt}( u(c_{t}) - h(a_{t}) )dt - r e^{-rt}W_{t}dt + e^{-rt}dW_{t}.
\end{equation}
Now, together \eqref{Agent's g-mgale} and \eqref{Agent's dV-path} imply that
\begin{equation*}
r e^{-rt}( u(c_{t}) - h(a_{t}) )dt - r e^{-rt}W_{t}dt + e^{-rt}dW_{t}
=\theta^{\ast}_{t} r e^{-rt}\kappa|Y_{t}|dt + \sigma r e^{-rt} Y_{t} dB^{a}_{t}
\end{equation*}
\begin{equation*}
\Longrightarrow \  dW_{t}= r \theta^{\ast}_{t} |Y_{t}|dt + \sigma r Y_{t} dB^{a}_{t}
 -r ( u(c_{t}) - h(a_{t}) )dt + r W_{t}dt.
\end{equation*}
The latter further implies
\begin{equation*}
W_{t}= W_{0} +
 \int_{0}^{t}r \left( W_{s} - u(c_{s}) + h(a_{s}) + \min_{\theta_{s} \in \Theta(a_{s})} \theta_{s}|Y_{s}| \right)ds +
 \int_{0}^{t} r Y_{s} dB^{a}_{s}
\end{equation*}
where $(B_t)$ with $B_t = X_t - a_t$ is a Brownian motion under the action strategy $\{a_{t}\}$. \qed

\subsection*{Proof of {Proposition}~\ref{prop:IC-kappa}}
This result follows as a special case of
the incentive compatibility condition for a slightly more general specification
$\Theta(\mu(a_t))$ for the drift-set.
The special case then arises by taking the drift set $\Theta(a)$ to be in $\kappa-$ignorance form:
$\Theta(a_t) = [a_t - \kappa(a_t), a_t + \kappa(a_t)]$.

\begin{prop}
\label{prop:IC-Ambig}
\normalfont
For a given strategy $a = \{a_{t}\}$, let $\{Y_{t}\}$ be the volatility process from
\normalfont{Proposition}~\ref{prop:Wdiff-Ambig-A}.
Then the action strategy $a$ is optimal if and only if
\begin{equation}\label{prop:ComparisonPA}
\forall a'_t \in \mathcal{A}, \
Y_{t}\mu(a_{t})- h(a_{t})   + \min_{\theta_{t} \in \Theta(a_t)}     \theta_{t}|Y_{t}|
\geq
Y_{t}\mu(a'_{t})- h(a'_{t}) + \min_{\theta_{t} \in \Theta(a'_t)}    \theta_{t}|Y_{t}|
\ \ dt \otimes dP \ a.e.
\end{equation}
\end{prop}

\vspace{-12mm}

\begin{proof}
Consider an arbitrary alternative strategy $a'= (a'_t)$ that follows possibly different actions
$a'_{\tau}$ up to $t$ and afterwards continues with $a_{t}$. The action strategy process $a'$
induces a set of densities $\Theta(a)$ satisfying the regularity conditions as specified earlier,
and since the set is IID (independent and indistinguishable in the sense of Chen and Epstein see Section
\ref{sect:set_of_priors} for details) the corresponding set of multiple priors $\mathcal{P}^{a'}$ is rectangular. The agent's expected payoff from this action strategy is well-defined by
\begin{equation*}
V'_{t} = \min_{Q \in \mathcal{P}^{\Theta(a')}} V_{t}^{Q},
\end{equation*}
where $V_{t}^{Q}$ is unique solution (ensured by \citet{duffie1992stochastic}) to BSDE
\begin{equation*}
V_{t}^{Q} = \expect^{Q} \left[ \int_{t}^{\infty} f(c_s,a'_s,V_{s}^{Q})\right],
\end{equation*}
where $f = u(c) - h(a) - \beta V $ is the  standard aggregator.
By $g-$martingale representation theorem of \citet[Theorem 2.2]{chen_epstein_ambiguity2002},
$V'$ is recursively represented in a unique manner:
\begin{equation*}
dV'_{t} = ( -f(c_{t},a'_{t},V'_{t}) + \max_{\theta \in \Theta({a'}_t)} \theta_{t} Y'_{t})dt
+ Y'_{t} dB^{a}_{t}
\end{equation*}
for a unique volatility process $Y'$.

More generally, $V'$ and $Y'$ uniquely solves a BSDE of the following form
\begin{equation}\label{BSDE}
d V_{t} = g'(V_t, Z_t, \omega, t)dt + Z_{t}dB^{a'}_t,
\end{equation}
with terminal condition $\xi$ for any path $\omega$ over the output realizations.
Specializing this result to the setup of in the current formulation we have
\begin{equation}\label{eqn:intertemporal utility}
g'(V,Z,\omega,t) = - f(c_{t}(\omega), a'_{t}(\omega),V)
                    + \max_{\theta \in \Theta(a'_t)} \theta(\omega)Y
\end{equation}
and the terminal condition is the value at the retirement.
It therefore follows that under the action strategy $\{a_{t}\}$ the agent's value process $V_{t}$ and its volatility $Y_{t}$ solve \eqref{BSDE}.

Suppose that the condition given in \eqref{prop:ComparisonPA} holds.
Since the terminal conditions corresponds to the retirement and at the retirement there is no uncertainty, the terminal conditions are invariant under different
action strategies $a$ and $a'$. By the Comparison Theorem in~\citet[Theorem 2.2]{el1997backward}
\begin{equation}\label{lemma:Comparison}
g(V,Z,\omega,t) \leq g'(V,Z,\omega,t) \quad dt \otimes dP \ a.e.
\end{equation}
or equivalent the condition~\eqref{prop:ComparisonPA} in our model holds
\begin{equation*}
\mu(a_{t})Y_{t} - h(a_{t}) -    \max_{\theta \in \Theta({a}_t)} \theta Y_{t}
\geq
\mu(a'_{t})Y_{t} - h(a'_{t}) -  \max_{\theta \in \Theta({a'_t})} \theta Y_{t}
\quad dt \otimes dP \ a.e.
\end{equation*}
which implies that $V \geq V'$ for almost every $t$.

Suppose now that the condition\eqref{prop:ComparisonPA} fails on a set of positive measures,
and choose an action strategy $a'$ that maximizes $\mu(\til[a])Y_{t} - h(\til[a]) - \max_{\theta \in \Theta({\til[a]})} \theta Y_t$ over $\til[a] \in \mathcal{A}$ for all $t \geq 0$. Then $g(V,Z,\omega,t) \leq g'(V,Z,\omega,t) \quad dt \otimes dP \ a.e.$. Since $a'$ specifies the same action as $a$ after $t$, an application of the comparison theorem implies $V' > V$. Therefore, $a$ is suboptimal.
\end{proof} 
If the volatility process is written as $-Y_{t}$ the minimum replaces
the maximum in \eqref{lemma:Comparison}.
Notice that \normalfont{Proposition}~\ref{prop:IC-Ambig} is formulated for
any generating process $\Theta({a})$.
Taking $\Theta({a}) := \{ (\theta_{t}): \mu(a_{t}) + |\theta_{t}| \leq \kappa(a_t) \}$
in the set up of \normalfont{Proposition}~\ref{prop:Wdiff-Ambig-A},
\normalfont{Proposition}~\ref{prop:IC-Ambig} now specializes
to the characterization of incentive compatibility $\kappa$-ignorance model:
\begin{lemma}\label{lemma:IC-kappa-A}
\normalfont
For any action strategy $a$, let $\{Y_{t}\}$ be the volatility process from
Proposition {\rm \ref{prop:Wdiff-Ambig-A}}. Then the strategy $a$ is incentive compatible
if and only if
\begin{equation*}
\forall \til[a]_t \in \mathcal{A} \quad
Y_{t}\mu(a_{t})- \kappa(a_{t}) |Y_{t}| - h(a_{t})  \geq
Y_{t}\mu(\til[a]_{t})- \kappa(\til[a]_{t}) |Y_{t}| - h(\til[a]_{t})
\quad dt \otimes dP \ a.e.
\end{equation*}
\end{lemma}
Finally, taking $\mu(a) = a$ for all $a \in \mathcal{A}$
specializes the result to \normalfont{Proposition}~\ref{prop:IC-kappa}. \qed

\subsection*{Proof of Theorem~\ref{thm:OptimalContract}}

%
%
%

In this proof, we  provide a {verification argument}
as a standard justification that
establishes validity  of
the dynamic programming approach to characterize
the principal's contracting problem recursively using
the agent's continuation value as a state variable.

We follow closely dynamic programming approach of
\citet{possamai2020there}
to analyze dynamic principal agent problems
in  a classical setting of  \citet{sannikov2008continuous},
we adopt it to allow for drift ambiguity.%
\footnote{
\citet{possamai2020there}
uses the general approach of \citet{lin2020random}
who develop formulation of the model
where the agent's action can control
 both the drift and the volatility of output process.}
In particular, the characterization of the optimal contract
as a solution to the HJB equation follows closely the proof strategy
used in the proof of {Theorem~3.6} and {Proposition~7.2} in \citet{possamai2020there}
for the case where the principal and the agent have a common discount factor.
Indeed, the main new mathematical arguments will be
due to possible heterogeneity in the worst-case drift scenarios
perceived by the principal and the agent.

We prove the verification argument
under a `regularity' assumption that
for any given promised value to the agent $W \geq 0$
the solution $(c(W),a(W),Y(W))$
of the HJB equation \eqref{eqn:HJBI-v}
 yields an admissible contract.
 In particular, this contract is defined
 by $\left( c_{t},a_{t} \right) := \left( c(W_{t}),a(W_{t}) \right)$
 with the associated incentive-compatible sensitivity is given
 $Y_t := Y(W_{t})$. For this contract,
 the agent's promised value process $\{ W_{t} \}$ evolves
 according to the flow diffusion equation
 \eqref{eqn:Ya}.
To streamline exposition,
we first prove the verification result under this regularity assumption.
Specifically, we show that under the regularity assumption
recursive optimal scheme $(c,a)$ is a solution to the principal's
sequential contracting problem. We then turn to verify this regularity property
by showing that recursive optimal scheme is indeed an admissible contract.
To simplify notation, we develop the verification argument
for any given drift-ambiguity scenario
$\phi$ and leave it implicit.

\textbf{Step 1.} Reproducing the contracting problem in sequential and recursive formulations

We re-state the sequential problem
\eqref{eqn:PrincipalsMinMax}
and denote by $v$ its value function defined by
\begin{equation}\label{eqn:PrincipalsMinMax-Appx}
V(W):=
\max_{(c,a) \in \mathcal{Z}(W_{0})}\min_{\theta \in \Theta({a})}
\expect^{a + \theta}
\left[r \int_{0}^{\infty} e^{-rt}( dX_{t} - c_{t}dt ) \right].
\tag{P}
\end{equation}
Here we recall that the set
$\mathcal{Z}(W_{0})$ describes the set of admissible of contracts
that satisfy incentive compatibility in \eqref{eqn:IC-sequential}
and integrability in \eqref{eqn:intableA} and \eqref{eqn:intableA-tau},
and yields to the agent ex-ante value $W_{0}$.

We also reproduce the recursive formulation of the contracting problem
and simplify the representation of its solution
\begin{equation}\label{eqn:HJBI-v}
\begin{split}
F(W;\phi) = \max_{(c,a,Y)\in \Gamma}
\min_{\theta \in \Theta(a;\phi)}
 \Big\{ r ( a + \theta - c )
 + r & F'(W;\phi) \left(W - u(c) + h(a) + \theta Y \right) \\
&+ \frac{F''(W;\phi)}{2} r^{2} Y^{2} \sigma^{2} \Big\}
\end{split}
\end{equation}
subject to the boundary conditions:
\[
F(0) = 0, \text{  and  } F \geq F_{0}.
\]
As in the retirement region the value function is equal to the value of retiring the agent,
i.e., $\mathcal{S}= \{F = F_{0}\}$,
the recursive formulation \eqref{eqn:HJBI-v} can be rewritten
in a compact form:
\begin{equation}\label{eqn:(7.1)}
\quad \quad  F(0)=0, \quad \text{   and   } \boldsymbol{L}(F) = 0 \text{  on  } W \in (0,\infty)
\end{equation}
where for any continuation value
$W>0$ the operator $\boldsymbol{L}$ is defined as follows:
\[
\begin{split}
\boldsymbol{L}F(W) := F  - \max_{(c,a,Y)\in \Gamma}
\min_{\theta \in \Theta(a)}
 \Big\{ r ( a + \theta - c )
 + & r  F'(W) \left(W - u(c) + h(a) + \theta Y \right) \\
&+ \frac{F''(W)}{2} r^{2} Y^{2} \sigma^{2} \Big\}
\end{split}
\]

\textbf{Step 2. } The recursive optimal scheme
yields a larger value than the principal's value function in the sequential problem: $F \geq V$.

Consider any admissible contract
$(c,a) \in Z(W_0)$ which  yields ex-ante value $W_0$ to the agent.
Let $\{W_{t}\}$ denote the corresponding
value process that describes the evolution of the agent's value from the contract
$(c,a)$,
according to the flow diffusion equation \eqref{eqn:Ya}.
Let  $\tau_{n} := \tau \wedge \inf\{t \geq 0: W_t \geq n\}$ a finite stopping time.
By Ito's formula, the principal's value
from the contract $(c,a)$ can be expressed as follows
{
\footnotesize
\begin{equation*}
\begin{split}
F(W_0) =& e^{-r \tau_{n}} F(W_{\tau_{n}}) - \int_{0}^{\tau_{n}}e^{-rt}
\{
-r F + \partial_{t}v + r F'(W_t) (W_{t} + u(c_{t}) - h(a_t) + \theta^{P}_{t} Y_{t} )
+ \frac{1}{2} \sigma^{2} r^{2} Y^{2}_{t} F''(W_t)
\}dt 	\\
&- \int_{0}^{\tau_{n}}e^{-rt} F'(W_t)r Y_{t} \sigma dB^{a + \theta^{P}}_{t}
\end{split}
\end{equation*}
}
As the solution to  HJB satisfies
$\boldsymbol{L}(F) = 0$ by \eqref{eqn:(7.1)}
and $F \geq F_{0}$ on the continuation region $\mathcal{S}^{c}$
the last equation implies
\[
\quad
\geq
e^{-r \tau_{n}} F_{0}(W_{\tau_{n}})
+ \int_{0}^{\tau_{n}}e^{-rt}
\left(
\boldsymbol{L}F(W_t) + a_t +\theta^{P}_t - c_t
\right)dt
- \int_{0}^{\tau_{n}}e^{-rt} F'(W_t)r Y_{t} \sigma dB^{a + \theta^{P}}_{t}
\]
\[
\quad
=
e^{-r \tau_{n}} F_{0}(W_{\tau_{n}})
+ \int_{0}^{\tau_{n}}e^{-rt}
\left( a_t +\theta^{P}_t - c_t
\right)dt
- \int_{0}^{\tau_{n}}e^{-rt} F'(W_t)r Y_{t} \sigma dB^{a + \theta^{P}}_{t}.
\]
As $F'$ is bounded on $[0,\tau_{n}]$,
and the admissible contract $(c,a)$ satisfies by assumption the integrability conditions
in \eqref{eqn:intableA}
and
\eqref{eqn:intableA-tau},
taking expectations in turn implies
\[
F(W_{0}) \geq \expect^{a + \theta^{P}}
\left[
e^{-r \tau_{n}} F_{0}(W_{\tau_{n}})
+ \int_{0}^{\tau_{n}}e^{-rt}
\left(
a_t +\theta^{P}_t - c_t
\right)dt
\right]
\]
\[
\quad  \quad \quad
    \xrightarrow{n\rightarrow\infty}
\expect^{a + \theta^{P}}
\left[
e^{-r \tau} F_{0}(W_{\tau})
+ \int_{0}^{\tau} e^{-rt}
\left(
a_t + \theta^{P}_{t} - c_t
\right)dt
\right]
\]
\[
\quad  \quad \quad
=
\min_{\theta \in \Theta({a})}\expect^{a+\theta}
\left[
e^{-r \tau_{n}} F_{0}(W_{\tau})
+ \int_{0}^{\tau_{n}}e^{-rt}
\left(
a_t + \theta_{t} - c_t
\right)dt
\right].
\]
Here the convergence as ${n\rightarrow\infty}$ follows from
integrability of the contract $(c,a)$ and the implied
continuation value process $\{W_{t}\}$.
Moreover, the equality holds as the term $\theta^{P}_{t}$
is the principal's worst-case scenario
under the contract $(c,a)$.
Finally, since the inequality holds for an arbitrary admissible contract $(c,a)$ that yields to the agent value $W_{0}$,  it implies that $F(W_0) \geq V(W_0)$.


\textbf{Step 3. } The principal's value function in the sequential problem
exceeds the value from the recursive optimal scheme: $F \geq V$.

For this, we start with considering
the solution $(c(W),a(W),Y(W))$
to the HJB for any admissible $W \geq 0$.
Associated to this solution,
define a recursive optimal scheme by
$(c_{t},a_{t}) := (c(W_{t}),a(W_{t}))$ together with sensitivity process
$Y_{t} := ( Y(W_{t}) )$ given by \eqref{eqn:IC-Y} for the action $a(W_{t})$
where the corresponding promised
value process $\{ W_{t} \}$ is defined by the flow diffusion equation \eqref{eqn:Ya}
starting with $W_{0}$ for this recursive optimal scheme.

For simplicity we assume that
 the solution $(c(W),a(W),Y(W))$ to the recursive equation HJB \eqref{eqn:HJBI} is unique.%
\footnote{The analogous arguments apply to general case of allowing for
multi-valued correspondence and working with each possible
value in the correspondence.}
Analogous to the argument as in the previous step,
letting $\tau_{n} := \tau \wedge \inf\{t \geq 0: W_t \geq n\}$ denote a finite stopping time
and applying Ito's formula to the principal's value yields
\[
F(W_0) =
\expect^{{a} + {\theta}^{P} }
\left[
e^{-r {\tau_{n}}} F({W}_{{\tau}_{n}})
+
\int_{0}^{{\tau}_{n}} e^{-rt} \left({a}_{t} + {\theta}^{P}_{t} - {c}_{t} \right)dt
\right]
\]
\[
\quad
    \xrightarrow{n\rightarrow\infty}
\expect^{{a} + {\theta}^{P} }
\left[
e^{-r {\tau}} F({W}_{{\tau}})
+
\int_{0}^{{\tau}} e^{-rt} \left({a}_{t} + {\theta}^{P}_{t}  - {c}_{t} \right)dt
\right]
\]
\[
\quad \quad \quad
=
\expect^{ a + {\theta}^{P} }
\left[
e^{-r{\tau}} F_{0}({W}_{{\tau}})
+
\int_{0}^{{\tau}} e^{-rt} \left({a}_{t}  + {\theta}^{P}_{t}  - {c}_{t} \right)dt
\right]
\leq V(W_{0})
\]
where we use the fact that in the retirement region $F = F_{0}$.
Here the  inequality holds as the recursive optimal scheme $(c,a) \in \boldsymbol{\mathcal{Z}}(W_{0})$
is an admissible contract and it yields the promised value $W_{0}$
to the agent. Since the recursive optimal contract is an admissible
contract in the sequential primal formulation,
the principal's optimization in the sequential formulation
over the (weakly) larger set of admissible contracts
yields him a (weakly) better payoff and hence $F(W_{0}) \leq V(W_{0})$.

\textbf{Step 4.} Verifying the regularity condition that the recursive optimal scheme
yields an admissible contract.

We have shown that the solution to the HJB
equation describes an optimal contract
under the regularity assumption that
the recursive optimal solution yields an admissible contract
$({c},{a}) \in \boldsymbol{\mathcal{Z}}(W_0)$.
In the remainder of this proof, we verify that
 this regularity property holds.

For this we start with noticing that
as the action space $\mathcal{A}$ and the consumption space $\mathcal{C}$
are compact subsets in $R_{+}$, respectively,
the mappings $a(W)$ and $c(W)$ are bounded on the continuation domain
$\boldsymbol{S}^{c}$.
Similarly, the optimal sensitivity $Y(W)$ implied by incentive compatibility
condition \eqref{eqn:IC-Y} is bounded on the continuation domain
$\boldsymbol{S}^{c}$.
Furthermore, as the value function
$F$ is continuous in $W$,  the objective function of HJB equation is continuous
and hence by the theorem of Maximum as an optimizer
the mapping ${Y(W)}$ is continuous on $\boldsymbol{S}^{c}$.
This therefore shows that rewriting
the agent's continuation value process
\eqref{eqn:Ya} as the following
\begin{equation}\label{eqn:5.6}
dW_{t} = r( W_{t} + h(a_{t}) - u(c_{t}) ) dt + r \sigma Y_{t} dB_{t}^{a+\theta^{A}}, \  P^{a+\theta^{A}} - a.s.
\end{equation}
where $dB_{t}^{a+\theta^{A}} = (a_{t} + \theta^{A}_{t})dt + dB_{t}$ is a Brownian motion under the measure $P^{a+\theta^{A}}$
for the agent's worst-case scenario,
the drift term is a bounded function of  $W_t$ and the volatility term is a continuous function of it.
Applying \citet{stroock1997multidimensional} (Corollary 6.4.4) now implies
that there exists a unique (weak) solution for the agent's continuation-value process $(W_{t})$ starting from $W_{0}$.

It remains to verify that the stopping-time
${\tau} := \inf\{t \geq 0: {W}_{t} \in \mathcal{S}\}$
 associated with the contract $(c,a)$
has a finite random horizon.
As the value process $\{W_{t} \}$ defined by \eqref{eqn:5.6} is a one-dimensional Markov
process for which the boundaries
$0$ and $W_{gp}$ are absorbing,
the stopping-time ${\tau}$ is finite
with probability $1$ under the distribution $P^{a+\theta^{A}}$
(see, for instance, \citet{helland1996one}).
Moreover,
as the action $\{a_{t}\}$ takes values in the compact set $\mathcal{A}$
and the drift $\theta^{A}(a_{t})$ for any $a_{t}$ belongs to a compact set $\Theta(a_{t})$,
Radon-Nikodym derivative of the distribution $P^{a+\theta^{A}}$
with respect to the reference measure
$dP^{a+\theta^{A}}/d P^{0}$ given by the Girsanov exponent in \eqref{eqn:densities}
 has finite moments of any order.
Therefore, the stopping-time $\tau$ is also finite almost surely under the measure $P^{0}$.
Similarly, for any alternative action process $\{\hat{a}_{t}\}$
to which the agent can deviate and the drift process
$\{\hat{\theta}_{t}\}$,
taking values in $\Theta(\hat{a}_{t})$,
the resulting process  $\hat{W}_{t}$ defined by \eqref{eqn:5.6}
using $dB_{t}^{\hat{a}+\hat{\theta}}$ is a Brownian motion
and hence a one-dimensional Markov process.
The stopping-time $\tau$ is therefore finite with probability $1$
under any such resulting distribution $P^{\hat{a}+\hat{\theta}}$.



\subsection*{Proof of {Proposition}~\ref{prop:common-worst-case}}
Recall that non-trivial incentives implies a positive variation in the continuation
payoff $Y > 0$ within the contracting relationship. This in turn implies that
the agent perceives the lower envelope to be the worst-case scenario for him:
$\theta^{A} = a - \kappa(a;\phi)$.

Using the characterization of the principal's worst-case scenario in \eqref{eqn:L},
we next show that whenever the principal's worst-case becomes upper envelope,
i.e., $a + \kappa(a;\phi)$, she would rather retire the agent, and hence during
the contracting relationship the lower-envelope is the principal's worst case.
In particular whenever upper-envelope is active for some $W'$ then  $W' > W_{gp}$.
Suppose to the contrary and hence for some $\tilde{W} < W'$,
and close to $W_{gp}$ by continuity of $F$ implies that
$F'(\tilde{W}) > F'_{o}(W_{gp}) = -\frac{1}{u'(C_{gp})} = -\frac{1}{\gamma_{o}}$
where $\gamma_{o}$ is the minimum of continuous function $Y(a) = h'(a)/(1 - \kappa'(a))$ over $a$ in a compact set $\mathcal{A}$.
Moreover, since the continuation value is in the interior of the continuation region,
$\tilde{W} \in (0,W_{gp})$, and $F(\tilde{W}) > F_{0}(\tilde{W})$,
taking linear approximation at $\tilde{W}$ yields
\[
F(\tilde{W}) + F'(\tilde{W})\tilde{W} > F_{o}(W_{gp}) + F'_{o}(W_{gp})W_{gp}
> -C_{gp} + F'(\tilde{W})u(C_{gp}).
\]

As $F'(\tilde{W}) > F'_{o}(W_{gp})$ and  the upper envelope in \eqref{eqn:L} is the worst case, the latter implies that
\[
\begin{array}{cl}
F''(\tilde{W}) & = \frac{F(\tilde{W}) - a -\kappa(a) + {c} - F'(\tilde{W})(\tilde{W} - u({c}) + h(a) + \kappa(a)Y)}{r \sigma^{2}Y^{2}/2} \\
& \geq \frac{F(\tilde{W}) - a -\kappa(a) + C_{gp} - F'(\tilde{W})(\tilde{W} - u(C_{gp}) + h(a) + \kappa(a)Y)}{r \sigma^{2}Y^{2}/2} \\
& >  \frac{-a +h(a)/\gamma_{o}}{r \sigma^{2}Y^{2}/2} > 0,
\end{array}
\]
contradicting concavity of the principal's value function. \qed

%
%


\subsection*{Proof of Proposition~\ref{prop:Basic}}
The proof of \ref{prop:kappamonotonicity} of  Proposition~\ref{prop:Basic}
follows from a 'revealed preference' argument.
Specifically, fix any initial  promised value $W$ and pick an optimal action strategy $( \widehat a_{t} )$ under $\widehat\phi$. Let $\widehat c^{\widehat a}$ and $c^{\widehat a}$ be the least costly wage schemes to implement these action strategies under $\widehat \phi$ and $\phi$, respectively.
Since $\widehat\phi> \phi$, the contract $(\widehat c^{\widehat a})$ together with action strategy
$(\widehat a_{t} )$ yields first-order stochastically better distribution over output paths under $\phi$ than under $\widehat \phi$
\[
F(W;\widehat \phi) =
\expect_{ \widehat a -  \kappa(\widehat a;\widehat \phi)}\left[\int_{0}^{\widehat\tau} X_{t} - \widehat c^{\widehat a}_{t} dt \right]
\leq \expect_{ \widehat a -  \kappa( \widehat a; \phi)}\left[\int_{0}^{\widehat\tau} X_{t} - \widehat c^{\widehat a}_{t} dt\right]
\label{eqn:K1} \tag{K1}
\]
Moreover,  since $c^{\widehat a}$ is the least costly way of implementing $\widehat a$ under $\phi$,
we have that:
\[
\expect_{\widehat a - \kappa(\widehat a;\phi)}\left[\int_{0}^{\widehat\tau} X_{t} - \widehat c^{\widehat a}_{t} dt \right]
\leq \expect_{\widehat a -  \kappa(\widehat a;\phi)}\left[\int_{0}^{\tau} X_{t} -  c^{\widehat a}_{t} dt\right]
\label{eqn:K2} \tag{K2}
\]
Finally, under $\phi$ the optimal action strategy $(a^{\ast}_t)$ does weakly better than $(\widehat a_{t})$:
\[
\expect_{\widehat a -  \kappa(\widehat a;\phi)}\left[\int_{0}^{\tau} X_{t} -  c^{\widehat a}_{t} dt\right]
\leq \expect_{a^{\ast} -  \kappa(a^{\ast};\phi)}\left[\int_{0}^{\tau} X_{t} -  c^{a^{\ast}}_{t} dt\right] = F(W;\phi)
\label{eqn:K3} \tag{K3}
\]
Together with \eqref{eqn:K1}-\eqref{eqn:K3} now gives $F(W;\phi) > F(W;\widehat \phi)$.
The latter in turn holds for all $W$ by Markov property of the solutions to HJBI equation~\eqref{eqn:HJBI}.

Since the profit function has the same initial condition and the same boundary conditions,
and the profit function under higher ambiguity is lower at all values of $W$.
The latter implies that the slope of the profit function must be lower under higher ambiguity at every continuation value $W$,
which shows the part~\ref{prop:kappaslopes}
of Proposition~\ref{prop:Basic}

To show the part \ref{prop:kappaconcavity} of Proposition~\ref{prop:Basic}
we reformulate HJBI in the following form:
\[
F''(W;\phi) = \min_{(c,a,Y)\in \Gamma} \frac{F(W) - (a - \kappa(a;\phi)) + c - F'(W)(W - u(c) + h(a) + \kappa(a;\phi) Y)}{r \sigma^{2}Y^{2}/2}
\]

By a generalized envelope theorem argument \citet{milgrom2002envelope}:
\begin{multline*}
F''_{\phi}(W;\phi) :=\frac{\partial F''(W;\phi)}{\partial \phi}
=
\frac{F_{\phi}(W;\phi)}{r \sigma^{2}Y^{2}/2}
+
\kappa_{\phi}(a;\phi)\frac{- 1 - F'(W;\phi)Y}{r \sigma^{2}Y^{2}/2}
\\
- \frac{F'_{\phi}(W;\phi)(W - u(c) + h(a) + \kappa(a;\phi) Y)}{r \sigma^{2}Y^{2}/2}
< 0
\end{multline*}
evaluated at the optimal controls $a,c$ and $Y$.\footnote{
The parameter $\phi$ does not affect the distribution $P$ of the underlying process, which justifies the application of the theorem.}
The first term is negative because $F_{\phi}(W;\phi) \leq 0$ by
the part~\ref{prop:kappamonotonicity} of Proposition~\ref{prop:Basic}; the second term is negative in view of  $- 1 - F'(W;\phi)Y < 0 $ -- property \eqref{eqn:L} that characterizes in Proposition~\ref{prop:common-worst-case}
the lower envelope as the common worst-case scenario, and
by Definition~\ref{defn:higher-ambig}
higher strength of ambiguity implies a larger drift-set $\kappa_{\phi}(a;\phi) > 0$
The third term is also negative because $F'_{\phi}(W;\phi) \leq 0$ by
the part~\ref{prop:kappaslopes} of Proposition~\ref{prop:Basic}
and the drift of the agent's continuation
value $(W - u(c) + h(a) + \kappa(a;\phi) Y)$ is non-negative as the zero-action
is always feasible. \qed

\subsection*{Proof of {Proposition}~\ref{thm:ContinuationDomain}}
Suppose that $W_t \in C(\phi',t)$. This implies that
$F(W_t,\phi') > F_0 (W_t)$. The part~\ref{prop:kappamonotonicity}
of Proposition~\ref{prop:Basic}
implies that $F(W_t,\phi) > F(W_t,\phi')$, and hence $W_t \in C(\phi)$. \qed


\subsection*{Proof of {Theorem}~\ref{thm:Flattening-Linears}}

The proof strategy implements monotone comparative statics for three kinds of affine
forms relating actions to the degrees of drift ambiguity:
constant, linear increasing, and linear decreasing.
We treat each in turn below.


\textbf{Constant Ambiguity.}
Consider a drift ambiguity model with
{$\kappa(a;\phi) = \phi_{0}$ for all $a$}.
With this formulation of the drift ambiguity, using the lower envelope $a - \phi_{0}$ as
the (common) worst-case scenario the functional equation
HJBI~\eqref{eqn:HJBI}
describing the principal's contracting problem
therefore takes the following form
\begin{equation}\label{eqn:HJB-phi-zero}
\begin{split}
r F(W;\phi_{0}) = \max_{(a,Y,c)\geq 0 } \big\{ r \left( a - \phi_{0} - c  \right)
+ r F'(W;\phi_{0}) \left(W \right. & - u(c)  \left. + h(a) - \phi  Y \right) \\
&+ {F''(W;\phi_{0})} r^{2} Y^{2} \sigma^{2}/2 \big\}
\end{split}
\end{equation}
subject to the incentive compatibility \eqref{eqn:IC-Y}
and boundary conditions in \eqref{eqn:HJBI}.

The incentive compatibility for the constant ambiguity model
implies the following relationship  between an action $a$ and incentive compatible variation in the continuation value
to the agent $Y$ required to implement it:
$Y= a$. It does not depend on the constant ambiguity
as the agent's action choice in this model does not affect it.
Using this relationship and letting
$m(Y;\phi_{0}) := \left(W - u(c) + h(Y) - \phi_{0} Y \right)$
denote the drift of the agent's continuation value (from the principal's perspective)
expresses the HJBI equation~\eqref{eqn:HJB-phi-zero}
in the following form:
\[
r F(W;\phi_{0}) = \max_{(Y,c)\geq 0 } \big\{ r \left( Y - \phi_{0} - c  \right)
+ r F'(W;\phi_{0}) m(Y;\phi_{0}) + {F''(W;\phi_{0})} r^{2} Y^{2} \sigma^{2}/2 \big\}
\]
subject to the same set of boundary conditions as in
\eqref{eqn:HJB-phi-zero}.

Interest is in the monotone comparative statics argument
for the optimal incentive $Y$ respect to constant (common) strength
of ambiguity $\phi_{0}$.
This is based on analyzing the nature of modularity for the objective function
with respect to $Y$ and $\phi_{0}$.
For this, we focus on the terms of the objective function that involve interactions between
the incentives $Y$ and strength of ambiguity $\phi_{0}$.
The partial derivative of the related terms
with respect to $\phi_{0}$:
\[
\mathcal{H}_{\phi_{0}}(Y) =
rF_{\phi}'(W;\phi_{0}) m(Y;\phi_{0})
 - r F'(W;\phi_{0}) Y + {F_{\phi_{0}}''(W;\phi_{0})} r^{2} Y^{2} \sigma^{2}/2
\]
We next analyze whether this partial derivative is decreasing(increasing) in $Y$
and determine whether the objective function is sub(super)-modular in $Y$ and $\phi_{0}$.
By envelope theorem applied to the agent's
optimization problem, the first term of the partial derivative $\mathcal{H}_{\phi_{0}}(Y)$
does not vary with $Y$.
The last term is also monotone decreasing in $Y$
as $F_{\phi_{0}}''(W;\phi_{0}) < 0$ by part \ref{prop:kappaconcavity} of
\textnormal{Proposition~\ref{prop:Basic}}.
The second term, on the other hand, is first
decreasing in $Y$ for low enough continuation values such that $W \in (W_o, \upbar{W})$
since on this set $F'(W;\phi_{0}) \geq 0$, and then it is
increasing in $Y$ for high enough continuation values
such that $W \in (\upbar{W},W_{gp} )$.
As the sum of these terms the partial derivative
$\mathcal{H}_{\phi_{0}}(Y)$
 is therefore monotone decreasing in $Y$
 and hence the objective function is sub-modular
 for low enough continuation values such that $W \in (W_o, \upbar{W})$.
 For high enough continuation values $W \in (\upbar{W},W_{gp})$ on the other hand
 it may be super-modular if super-modularity due to the second term overcomes
the sub-modularity due to the other terms.

Summarizing, the characterization of monotone comparative statics
for the optimal incentive $Y$ with respect to the constant ambiguity $\phi_{0}$
is as follows.
For low enough continuation values $W \in (W_o, \upbar{W})$,
the objective function in \eqref{eqn:HJB-phi-zero}
is sub-modular in $(Y,\phi_{0})$ and by the monotone comparative static argument,
the optimal incentive sensitivity $Y$
decreases in $\phi_{0}$.
For high enough continuation values $W \in (\upbar{W}, W_{gp})$, on the other hand,
the objective function in \eqref{eqn:HJB-phi-zero}
can be super-modular in $(Y,\phi_{0})$ and
by the monotone comparative static argument,
hence the optimal sensitivity $Y$
can be non-decreasing in $\phi_{0}$.
This yields the result in the part (a) of \normalfont{Theorem}~\ref{thm:Flattening-Linears}.
We next turn to analyze how the optimal $Y$ changes
as the slope term $\phi_{1}$ changes.

 \textbf{Ambiguity linearly increasing in action.}
Consider a linear drift ambiguity model with $\kappa(a;\phi_{1}) = \phi_{1} a$.
With this formulation of the drift ambiguity, using the lower envelope $a (1- \phi)$ as
the (common) worst-case scenario
the HJBI~\eqref{eqn:HJBI}
therefore takes the  form:
\begin{equation}\label{eqn:HJBI-Lin-zero}
\begin{split}
\hspace{-2mm} r F(W;\phi_{1}) = \max_{(a,Y,c)\geq 0 } \big\{ r \left( a(1 - \phi_{1}) - c  \right)
+ & r F'(W;\phi_{1}) \left(W \right. - u(c)  \left. + h(a) - \phi_{1} a Y \right) \\
&+ {F''(W;\phi_{1})} r^{2} Y^{2} \sigma^{2}/2 \big\}
\end{split}
\end{equation}
subject to the boundary conditions (corresponding to individual rationality requirements), and
incentive compatibility \eqref{eqn:IC-Y}.
The incentive compatibility  implies the following relationship
between an action $a$ and incentive compatible variation in the continuation value
to the agent $Y$ required to implement it:
$Y= {a}/{1 - \phi_{1}}$.
Using this relationship so that $a(Y;\phi_{1}) = (1 - \phi_{1})Y$
and letting $m(Y;\phi_{1}) = \left(W - u(c) + h(a(Y;\phi)) - \phi_{1} a(Y;\phi) Y \right)$
denote the drift of the agent's continuation value (from the principal's perspective)
expresses the HJBI equation~\eqref{eqn:HJBI-Lin-zero}
in the following form:
\[
r F(W;\phi_{1}) = \max_{(Y,c)\geq 0 } \big\{ r \left( a(Y;\phi_{1})(1 - \phi_{1}) - c  \right)
+ r F'(W;\phi_{1}) m(Y;\phi_{1}) + {F''(W;\phi_{1})} r^{2} Y^{2} \sigma^{2}/2 \big\}
\]
subject to the same set of boundary conditions as in
\eqref{eqn:HJBI-Lin-zero}.
For monotone comparative statics argument,
we focus on the terms of the objective function that involve
interactions between incentives $Y$ and
the strength of ambiguity $\phi_{1}$
\begin{equation}\label{eqn:Hfull}
\mathcal{H} (Y;\phi_{1}):=
r  (1 - \phi_{1})^{2} Y
+ r F'(W;\phi_{1}) m (Y;\phi_{1}) + {F''(W;\phi_{1})} r^{2} Y^{2} \sigma^{2}/2
\end{equation}
We next analyze modularity of this function in terms of
$Y$ and $\phi_{1}$ which in turn determines the monotone comparative statics
in the optimal contract for $Y$ in terms of $\phi$.
The partial derivative of $\mathcal{H} (Y;\phi)$ with respect to $\phi$  is
\[
\mathcal{H}_{\phi_{1}}(Y) :=
- r 2(1 - \phi_{1})Y
+ F_{\phi_{1}}'(W;\phi) m(Y;\phi_{1})
+ F'(W;\phi_{1}) m_{\phi} (Y;\phi_{1})
+ {F_{\phi_{1}}''(W;\phi_{1})} r Y^{2} \sigma^{2}/2
\]
The first term of this partial derivative is decreasing in $Y$ since $\phi_{1} \leq 1$
for any viable contracting relationship for otherwise the incentive compatibility condition
 \eqref{eqn:IC-Y}
implies that the optimal contract would implement zero-action corresponding to termination.
The second term at the optimum contract
does not vary in  $Y$ by the envelope theorem
applied to the agent's optimization problem.
Turning to the third term of the partial derivative,
since $m_{\phi_{1}} (Y;\phi_{1})$ is decreasing in $Y$,
it is  decreasing in $Y$ for low enough continuation value so that
$W \in (W_o, \upbar{W})$ as in this range $F'(W;\phi_{1}) \geq 0$.
On the other hand, it is then increasing in $Y$ for high enough continuation values
$W \in (\upbar{W},W_{gp})$ for which $F'(W;\phi_{1}) < 0$.
Finally, the last term of the partial derivative $\mathcal{H}_{\phi_{1}}(Y)$
is monotone decreasing in $Y$ as $F_{\phi_{1}}'' < 0$ by the part \ref{prop:kappaconcavity}
of Proposition~\ref{prop:Basic}.
Considering the sum of these terms then,
except possibly for high enough continuation values such that
$W \in (\upbar{W},W_{gp})$
the objective function is sub-modular in
$Y$ and $\phi_{1}$ in the linear model,
and hence
by monotone comparative statics arguments
the optimal incentives $Y$ decreases in the principal's linear ambiguity $\phi_{1}$.

\textbf{Ambiguity linearly decreasing in effort}:
In this formulation, the length of drift-set is given by
$\kappa(a;\bphi) = {\phi_{0}} - \phi_{1} a$.
With this linear decreasing formulation of the drift ambiguity,
using the lower envelope $a - \kappa(a;\bphi)$
as the (common) worst-case scenario in drift-sets
the functional equation HJBI~\eqref{eqn:HJBI} describing the principal's contracting
problem therefore takes the  form:
\[
\begin{split}
r F(W;\bphi) = \max_{(a,Y,c)\geq 0 } \big\{ r ( a - \kappa(a;\bphi)  - c  )
+ r F'(W;\bphi) \left(W \right. &- u(c)  \left. + h(a) - \kappa(a;\bphi) Y \right) \\
&+ {F''(W;\bphi)} r^{2} Y^{2} \sigma^{2}/2 \big\}
\end{split}
\]
subject to boundary conditions in \eqref{eqn:HJBI}, and
incentive compatibility \eqref{eqn:IC-Y}.
The latter  for the linear decreasing model of drift set
implies the following relationship  between an action $a$ and sensitivity in the continuation value to the agent $Y$ to implement it:
$Y= {a}/{1 + \phi_{1}}$.
Using this relationship $a(Y;\bphi) := (1 + \phi_{1})Y$
and letting $m(Y;\bphi) = \left(W - u(c) + h(a(Y;\bphi)) -  (\phi_{0} - \phi_{1} a(Y;\bphi)) Y \right)$
denote the drift of the agent's continuation value (from the principal's perspective)
expresses the objective function in the following form:
\[
r F(W;\bphi) = \max_{(Y,c)\geq 0 } \big\{ r \left( a(Y;\bphi)(1 + \phi_{1}) - \phi_{0} - c  \right)
+ r F'(W;\bphi) m(Y;\bphi) + {F''(W;\bphi)} r^{2} Y^{2} \sigma^{2}/2 \big\}
\]
Notice that in the linear decreasing formulation
constant ambiguity $\phi_{0}$  and slope ambiguity $\phi_{1}$
enters into the formulation linearly additivity.
Due to this separability, the monotone comparative static analysis of $Y$
with respect to the constant ambiguity $\phi_{0}$
follows by the arguments analogous to that made earlier in this proof
for the constant ambiguity model.

The monotone comparative static analysis with respect to $\phi_{1}$
follows from analogous arguments made above for
linear increasing model by a change of variable.
This change of variable replaces
$\phi_{1}$ in the linear increasing ambiguity model
with $-\phi_{1}$ to represent drift set becoming smaller as effort increases.
In this specification, higher drift ambiguity
arises from an increase in $\phi_{0}$ and a decrease in $\phi_{1}$,
which expand drift sets.
From this change of variables, the analogous characterization follows.
\vspace{-5mm}

\subsection*{Proof of Proposition~\ref{prop:Wage-Depressing}}


The optimal wage is determined by solving the principal's optimization problem~\eqref{eqn:H-Principtheta-Y},
and it solves the first-order condition:
\begin{equation*}
1 = F'(W;\phi) u'(c(W;\phi))
\end{equation*}
Since  $F'_\phi < 0 $ by
the part~\ref{prop:kappaslopes} of Proposition~\ref{prop:Basic} and  $u$ is concave,
this equation implies that $c(W;\widehat\phi) < c(W;\phi)$ whenever
$\widehat\phi \succ_{\Amb} \phi$.
Moreover,
as the principal's profits are lower at higher ambiguity levels $F(W;\widehat\phi) \leq F(W;\phi)$,
by Proposition~\ref{prop:Basic}, whenever the two contracts specify the same wage levels
for the associated continuation values $\widehat W$ and $W$, respectively,
by analogous arguments in \citet[Theorem 4]{sannikov2008continuous}
the ratio of volatilities of the agent's consumption and continuation values is greater
in the contract associated with $\widehat\phi$.
Therefore, the optimal contract associated with higher ambiguity $\widehat\phi$
relies less on short-term incentives. \qed

\vspace{-5mm}

\section*{Proofs of the Results for Heterogeneous Drift Ambiguity in Section~\ref{Section:Heterogeneous Drift sets}}

\subsection*{Proof of Proposition~\ref{prop:Basic-Asy}}

%
The proof of this proposition extends and applies the arguments made in
 \normalfont{Proposition}~\ref{prop:Basic} which characterizes the basic properties
 of the principal's value function for the symmetric drift ambiguity.
The proof of the part \ref{prop:kappamonotonicity-Asy}
Proposition~\ref{prop:Basic-Asy}
for heterogeneous
ambiguity follows  from a `revealed preference' argument
analogous to that made for symmetric case in the part \ref{prop:kappamonotonicity}
of \normalfont{Proposition}~\ref{prop:Basic}.

 Consider first a change  in the ambiguous perception of the principal. Specially, suppose that starting from the ambiguity scenario characterized by
$\bphi= (\phi^{P},\phi^{A})$
the ambiguity perceived by the principal increases from $\phi^{P}$ to $\widehat\phi^{P}$
i.e., $\kappa(a;\widehat\phi^{P}) > \kappa(a;\phi^{P})$ for each $a \in \actions$
while the agent's ambiguous perception remains unchanged at $\phi^{A}$. Denote this new scenario
by $\widehat{\bphi} := (\widehat\phi^{P},\phi^{A})$.

For a `revealed preference' argument, fix any promised value $W$ and let
$(\widehat c, \widehat a)$ be the optimal contract under the ambiguity scenario $\widehat{\bphi}$
that yields value $W$ to the agent.
As $\kappa(a;\widehat\phi^{P})\geq \kappa(a;\phi^{P})$,
the lower envelope being the worst-case scenario for the principal,
implies that the action profile $(\widehat a_{t} )$ yields a better worst-case distribution in
first-order stochastic order under $\phi^{P}$ than under $\widehat \phi^{P}$:
\begin{multline*}
F(W;\widehat \bphi) :=
\expect_{P} \left[ \int_{0}^{\widehat \tau}e^{-r t}( \widehat a_t - \kappa( \widehat a_t;\widehat\phi^{P}) - \widehat c_t )dt \right]
\\
\leq
\expect_{P} \left[ \int_{0}^{\widehat \tau}e^{-r t}( \widehat a_t - \kappa( \widehat a_t;\phi^{P}) - \widehat c_t )dt \right]. \label{eqn:K1-Asy} \tag{K1$'$}
\end{multline*}
Here the inequality is strict if $\kappa(a;\widehat\phi^{P}) > \kappa(a;\phi^{P})$
for each $a \in \actions$.

Now consider the effect of  the agent's ambiguity perceptions and
that it increases from $\phi^{A}$ to $\widehat \phi^{A}$
such that $\kappa(a;\widehat \phi^{A}) \geq \kappa(a;\phi^{A})$ for each $a \in \actions$.
It implies from  \eqref{eqn:IC-Y-Asy} that the agent's action choice increases
as his perceived ambiguity decreases as it reduces the agency cost of implementation.
In particular,  denote by $\{a_{t}\}$ the action profile that the
wage scheme $(\widehat c)$, which is optimal under the scenario $\widehat \phi$,
with the associated sensitivity process $(\widehat Y_t)$ implements under $\phi^{A}$.
Notice from the incentive compatibility condition \eqref{eqn:IC-Y-Asy}
this incentive scheme implements a higher action profile
under $\phi^{A}$ than $\widehat\phi^{A}$, or more precisely:
$\widehat \phi^{A} \geq \phi^{A}$ implies $a_t \geq \widehat a_t$,
with strict inequality if $ \widehat \phi^{A}  > \phi^{A} $.
This in turn implies for the principal's guaranteed payoff
\[
\expect_{P}\left[\int_{0}^{\widehat\tau}e^{-r t} ( \widehat a_t - \kappa( \widehat a_t;\widehat\phi^{P}) - \widehat c_t) dt \right]
\leq
\expect_{P}\left[\int_{0}^{\widehat \tau}e^{-r t}( a_t - \kappa(  a_t;\widehat\phi^{P}) - \widehat c_t)dt\right]
\label{eqn:K2-Asy} \tag{K2$'$}
\]
with strict inequality if
$\kappa(a;\widehat \phi^{A}) > \kappa(a;\phi^{A})$ for each $a \in \actions$.

As the principal's optimal contract chooses optimally amongst the incentive
compatible allocations, under ambiguity scenario $\bphi$
the optimal action strategy $(a^{\ast}_t)$ does weakly better for the principal
than the action strategy $(\widehat a_{t})$, which is optimal under $\widehat\bphi$:
\[
\expect_{P} \left[ \int_{0}^{\widehat \tau}e^{-r t}(\widehat a_t - \kappa(\widehat a_t;\phi^{P}) - \widehat c_t) dt \right]
\leq
\expect_{P} \left [ \int_{0}^{\tau^*}e^{-r t}(a^{\ast}_t - \kappa( a^{\ast}_t;\phi^{P}) - c^{\ast}_t) dt \right ] := F(W;\widehat\phi)~\tag{K3$'$}
\label{eqn:K3-Asy}
\]
Summarizing, together with \eqref{eqn:K1-Asy}-\eqref{eqn:K3-Asy} implies
that $F(W;\bphi) > F(W;\widehat\bphi)$ whenever $(\phi^{P},\phi^{A}) \succ_A (\widehat \phi^{P},\widehat \phi^{A})$.
The latter in turn holds for all $W$ by Markov property of the solutions to
the HJBI equation~\eqref{eqn:HJBI-Asy}, which shows
the part \ref{prop:kappamonotonicity-Asy} of
Proposition~\ref{prop:Basic-Asy}.

To see that the part \ref{prop:kappaslopes-Asy} holds notice that
the profit function has the same initial conditions and the same boundary conditions
both under $\bphi$ and $\widehat\bphi$.
This together with the part~\ref{prop:kappamonotonicity-Asy} implies that the slope of
the profit function must be lower under higher ambiguity.

To show the part~\ref{prop:kappaconcavity-Asy}
of Proposition~\ref{prop:Basic-Asy}
  we start with reformulating HJBI~\ref{eqn:HJBI-Asy}:
\begin{equation*}
F''(W;\bphi) = \min_{(a,Y,c)\in \Gamma} \frac{F(W;\bphi) - (a - \kappa(a;\phi^{P})) + c - F'(W;\bphi)(W - u(c) + h(a) -\kappa(a;\phi^{A}) Y)}{r \sigma^{2}Y^{2}/2}
\end{equation*}
and then analyze the effects of changes in the principal's ambiguity $\phi^{P}$ and
the agent's ambiguity $\phi^{A}$ on the concavity of the principal's value function, in turn.
By a generalized envelope theorem argument \citet{milgrom2002envelope}
its derivative with respect to $\phi^{P}$:
\begin{multline*}
F''_{\phi^{P}}(W;\bphi) :=\frac{\partial F''(W;\bphi)}{\partial \phi^{P}}
=
\frac{F_{\phi^{P}}(W;\bphi)}{r \sigma^{2}Y^{2}/2}
- \kappa_{\phi^{P}}(a;\phi^{A})\frac{1 + F'(W;\bphi)Y}{r \sigma^{2}Y^{2}/2}
\\
+ \frac{F'_{\phi^{P}}(W;\bphi)(W - u(c) + h(a) - \kappa(a;\phi^{A}) Y)}{r \sigma^{2}Y^{2}/2}
< 0
\end{multline*}
evaluated at the optimal controls $a,c$ and $Y$.
The first term is negative as $F_{\phi^{P}}(W;\bphi) < 0 $ by \ref{prop:kappamonotonicity-Asy}.
The second term is zero since the agent's action choice depends on
her perceived ambiguity $\phi^{A}$ but not on the principal's $\phi^{P}$, i.e.,
$\kappa_{\phi^{P}}(a;\phi^{A}) = 0$. The third and
the last term is also negative because $F'_{\phi^{P}}(W;\bphi) \leq 0$ by
the part~\ref{prop:kappaslopes-Asy}
of Proposition~\ref{prop:Basic-Asy}
in the current proof
and the drift of the agent's continuation
value, $(W - u(c) + h(a) - \kappa(a;\phi^{P}) Y)$, is non-negative during
the optimal contract as the zero-action is always feasible to implement.

Again by the generalized envelope theorem its derivative
this time with respect to $\phi^{A}$:

\begin{multline*}
F''_{\phi^{A}}(W;\bphi) :=\frac{\partial F''(W;\bphi)}{\partial \phi^{A}}
=
\frac{F_{\phi^{A}}(W;\bphi)}{r \sigma^{2}Y^{2}/2}
- \kappa_{\phi^{A}}(a;\phi^{A})\frac{1 + F'(W;\bphi)Y}{r \sigma^{2}Y^{2}/2}
\\
+ \frac{F'_{\phi^{A}}(W;\bphi)(W - u(c) + h(a) - \kappa(a;\bphi) Y)}{r \sigma^{2}Y^{2}/2}
< 0
\end{multline*}
The first term is negative as $F_{\phi^{A}}(W;\bphi) < 0 $ by the part \ref{prop:kappamonotonicity-Asy} of Proposition~\ref{prop:Basic-Asy}.
The second term is also negative  in view of  $ 1 + F'(W;\bphi)Y > 0 $ --  condition~\eqref{eqn:L'}  that characterizes the lower envelope as the principal's worst-case scenario, and
by \normalfont{Definition}~\ref{defn:higher-ambig} that
higher strength of ambiguity implies a larger drift-set
$\kappa_{\phi^{A}}(a;\phi^{A}) > 0$ for each $a \in \actions$. The third and
the last term is also negative because $F'_{\phi^{A}}(W;\bphi) \leq 0$ by
the part~\ref{prop:kappaslopes-Asy} of Proposition~\ref{prop:Basic-Asy}
and the drift of the agent's continuation
value is non-negative for any non-trivial action. \qed

\subsection*{Monotone Comparative Statics for  Heterogeneous Drift-Ambiguity in
Section~ \ref{subsection:Monotonce-Het}
}

\subsection*{Proof of {Theorem}~\ref{thm:Flattening-Linears-Het}}
The proof strategy
using monotone comparative statics for heterogeneous drift-ambiguity
that in \normalfont{Theorem}~\ref{thm:Flattening-Linears}.
We first analyze the case of constant ambiguity independent of action,
and then turn to the case that specify
linear relationship between actions and drift sets.

\textbf{Constant Ambiguity.}
For the constant ambiguity model,
notice that the incentive compatibility condition~
\ref{eqn:IC-Y-Asy} does not depend on the strength of (constant) ambiguity,
and implies that the required sensitivity to implement a non-trivial
action  $a$ satisfies $ Y(a;\phi^{A}_{0})= a$.
Using this relationship, notice that
the drift of the agent's continuation value
from the principal's perspective is
given by $m(Y;\phi^{P}_{0}) := \left(W - u(c) + h(Y) - \phi^{P}_{0} Y \right)$,
depends only on the principal's ambiguity
perception $\phi^{P}_{0}$ but not on the agent's $\phi^{A}_{0}$.
The HJBI formulation~\eqref{eqn:HJBI-Asy} of the
contracting problem therefore takes the form:
\begin{equation}\label{eqn:HJBI-Asy-Const}
\hspace{-5mm}
r F(W;\bphi_{0}) = \max_{(Y,c)\geq 0 } \Big\{ r \left( Y - \phi^{P}_{0} - c  \right)
+ r F'(W;\bphi_{0}) m(Y;\phi^{P}_{0}) + \frac{F''(W;\bphi_{0})}{2} r^{2} Y^{2} \sigma^{2} \Big\}
\end{equation}
subject to the same boundary and smooth-pasting conditions.
The monotone comparative statics argument,
is based on analyzing the nature of modularity for the objective function
with respect to
the incentives $Y$ and strength of ambiguity $\bphi_{0}$.
For this, we focus on the terms of the objective function that involve interactions between
$Y$ and $\bphi_{0}$.
We analyze the effects of $\phi^{P}_{0}$ and $\phi^{A}_{0}$
on the optimal incentive scheme in turn.
The partial derivative of the related terms
with respect to $\phi^{P}_{0}$:
\[
\mathcal{H}_{\phi^{P}_{0}}(Y) =
rF_{\phi^{P}_{0}}'(W;\bphi) m(Y;\phi^{P}_{0})
 - r F'(W;\bphi_{0}) Y + {F_{\phi^{P}_{0}}''(W;\bphi_{0})} r^{2} Y^{2} \sigma^{2}/2.
\]
We next analyze whether this partial derivative is decreasing(increasing) in $Y$
and hence
determine whether the objective function is sub(super)-modular in $Y$ and $\bphi_{0}$.
The first term of the partial derivative $\mathcal{H}_{\phi^{P}_{0}}(Y)$
is monotone decreasing in $Y$.
This is because during the contracting relationship
$F_{\phi^{P}_{0}}'(W;\bphi_{0}) < 0$
by the part~\ref{prop:kappaslopes-Asy} of
\textnormal{Proposition~\ref{prop:Basic-Asy}}
and the drift of the agent's continuation value $m(Y;\bphi_{0})$
is increasing in $Y$ for any viable action, i.e., $a - \phi^{P}_{0} > 0$,
for otherwise the principal would optimally implement zero action.
The last term is also monotone decreasing in $Y$
as $F_{\phi^{P}_{0}}''(W;\bphi_{0}) < 0$ by the part~\ref{prop:kappaconcavity-Asy} of
\textnormal{Proposition~\ref{prop:Basic-Asy}}.
The second term, on the other hand, is first
decreasing in $Y$ for low enough continuation values such that $W \in (W_o, \upbar{W})$
since on this set $F'(W;\bphi_{0}) \geq 0$, and then it is
increasing in $Y$ for high enough continuation values
such that $W \in (\upbar{W},W_{gp} )$.
The sum is therefore monotone decreasing in $Y$
 and hence the objective function is sub-modular
 for low enough continuation values $W \in (W_o, \upbar{W})$.
 For high continuation values $W \in (\upbar{W},W_{gp})$ on the other hand
 it may be super-modular if super-modularity due to the second term overcomes
the sub-modularity due to the other two terms.

Summarizing then gives the following characterization.
 For low enough continuation values $W \in (W_o, \upbar{W})$,
the objective function in \eqref{eqn:HJBI-Asy-Const}
is sub-modular in $(Y,\phi^{P}_{0})$ and the optimal incentive sensitivity $Y$
decreases in $\phi^{P}_{0}$. On the other hand,
for high enough continuation values $W \in (\upbar{W}, W_{gp})$,
the objective function can be super-modular in $(Y,\phi^{P}_{0})$ and hence the optimal sensitivity $Y$
can be non-decreasing in $\phi^{P}_{0}$.
This yields the result stated in the part (a) of \normalfont{Theorem}~\ref{thm:Flattening-Linears-Het} for an increase in $\phi_{0}^{P}$.

We next turn to the comparative statics of $Y$ with respect to  $\phi^{A}_{0}$.
For this, we focus on the terms of the objective function
\eqref{eqn:HJBI-Asy-Const} involving interaction
between them and consider
its partial derivative with respect to $\phi^{A}_{0}$
\[
\mathcal{H}_{\phi^{A}_{0}}(Y) =
r F_{\phi^{A}_{0}}'(W;\bphi_{0}) m(Y;\phi^{P}_{0}) + {F_{\phi^{A}_{0}}''(W;\bphi_{0})} r^{2} Y^{2} \sigma^{2}/2
\]
This partial derivative is monotone decreasing in $Y$.  This is because $F_{\phi^{A}_{0}}'(W;\bphi_{0})  < 0$ and $F_{\phi^{A}_{0}}''(W;\bphi_{0})  < 0$
by the parts~\ref{prop:kappaslopes-Asy} and \ref{prop:kappaconcavity-Asy} of
\normalfont{Proposition}~\ref{prop:Basic-Asy}, respectively,
and the drift of the continuation value $m(Y;\bphi_{0}) $ is increasing in $Y$.
The objective function of HJBI~\ref{eqn:HJBI-Asy-Const}
is therefore sub-modular in $Y$ and $\phi^{A}_{0}$.
By monotone comparative statics argument then,
the optimal incentive sensitivity $Y$ decreases
with the agent's constant ambiguity $\phi^{A}_{0}$.
This yields the result stated in the part (b) of \normalfont{Theorem}~\ref{thm:Flattening-Linears-Het} for an increase in $\phi_{0}^{A}$.


\textbf{Ambiguity linearly increasing in action.}
Consider a linear drift ambiguity model with $\kappa(a;\phi^{P}_{1}) = \phi^{P}_{1} a$ and
$\kappa(a;\phi^{A}_{1}) = \phi^{A}_{1} a$ for all $a \in \actions$.
Denote this heterogenous ambiguity by $\bphi_{1} = (\phi^{P}_{1},\phi^{A}_{1})$.
With this formulation of the drift ambiguity, the HJBI formulation~\eqref{eqn:HJBI-Asy} of the contracting problem now takes the form:
\[
\begin{split}
r F(W;\bphi_{1}) = \max_{(a,Y,c)\geq 0 } \ \  \big\{ r \left( a(1 - \phi^{P}_{1}) - c  \right)
+ &r   F'(W;\bphi_{1}) \left(W  - u(c)  + h(a) - \phi^{P}_{1} a Y \right) \\
&+ {F''(W;\bphi_{1})} r^{2} Y^{2} \sigma^{2}/2 \big\}
\end{split}
\]
subject to boundary conditions and
incentive compatibility \eqref{eqn:IC-Y-Asy}.
The latter implies the following relationship
between an action $a$ and incentive compatible variation in the continuation value
to the agent $Y$ required to implement it:
$Y= \frac{a}{1 - \phi^{A}_{1}}$. Notice that this condition depends only on the agent's perception  $\phi^{A}_{1}$ but not the principal's $\phi^{P}_{1}$ as the incentive compatibility condition describes the agent's decision rule for his action choice for any given $Y$.
Using this relationship so that $a(Y;\bphi_{1}) = (1 - \phi^{A})Y$
and letting $m(Y;\bphi) = \left(W - u(c) + h(a(Y;\bphi)) - \phi^{P}_{1} a(Y;\bphi) Y \right)$
denote the drift of the agent's continuation value (from the principal's perspective)
now expresses the HJBI equation in the following form:
\begin{equation}\label{eqn:HJBI-Asy-Lin}
\begin{split}
r F(W;\bphi_{1}) = \max_{(Y,c)\geq 0 } \big\{  r \left( a(Y;\bphi_{1})(1 - \phi^{P}_{1}) - c  \right)
+ &r F'(W;\bphi_{1})  m(Y;\bphi_{1}) \\
&+ {F''(W;\bphi_{1})} r^{2} Y^{2} \sigma^{2}/2 \big\}
\end{split}
\end{equation}
subject to the same set of boundary conditions.

Analogous to monotone comparative statics analysis in the proof
of \normalfont{Theorem}~\ref{thm:Flattening-Linears}
 the focus is on the terms of the objective function of
HJBI functional equation \eqref{eqn:HJBI-Asy-Lin} that involve
interactions between the strength of incentives $Y$ and the strength of ambiguity
$\bphi_{1} = (\phi^{A}_{1},\phi^{P}_{1})$
\begin{equation}\label{eqn:Hfull-asy}
\mathcal{H} (Y;\bphi_{1}):=
r  (1 - \phi^{A}_{1})(1 - \phi^{P}_{1}) Y
+ r F'(W;\bphi_{1}) m (Y;\bphi_{1}) + {F''(W;\bphi_{1})} r^{2} Y^{2} \sigma^{2}/2
\end{equation}
We next analyze modularity of this function to characterize
the monotone comparative statics of the optimal
incentive sensitivity $Y$ with respect to $\bphi_{1}$.
We do so first with respect to the principal's ambiguity $\phi^{P}_{1}$ and
then turn to the agent's linear ambiguity $\phi^{A}_{1}$.
%
%
The partial derivative of $\mathcal{H} (Y;\bphi_{1})$ with respect to $\phi^{P}_{1}$ (up to positive constant $r$)  is
\[
\small{
\mathcal{H}_{\phi^{P}_{1}} (Y;\bphi_{1}) =
- (1 - \phi^{A}_{1})Y
+ F_{\phi^{P}_{1}}'(W;\bphi_{1}) m(Y;\bphi_{1})
+ F'(W;\bphi_{1}) m_{\phi^{P}_{1}} (Y;\bphi_{1})
+ {F_{\phi^{P}_{1}}''(W;\bphi_{1})} r Y^{2} \sigma^{2}/2
}
\]
The first term is decreasing in $Y$ because
in any viable contracting relationship $\phi^{A}_{1}\leq 1$, for otherwise the incentive compatibility condition
 \eqref{eqn:IC-Y-Asy}
implies that the optimal contract would implement zero-action corresponding to termination.
By envelope theorem applied to the agent's
decision problem,
the agent's continuation value at the optimum
does not vary with differential change in $Y$.
Turning to the third term,
since $m_{\phi^{P}_{1}} (Y;\bphi_{1}) = -(1 - \phi^{A}_{1})Y^{2}$ is decreasing in $Y$,
it is decreasing in $Y$ for low enough continuation values so that
$W \in (W_o, \upbar{W})$ as in this range $F'(W;\bphi_{1}) \geq 0$.
On the other hand, it is then increasing in $Y$ for high enough continuation values
$W \in (\upbar{W},W_{gp})$ for which $F'(W;\bphi_{1}) < 0$.
Finally, the last term  of  $\mathcal{H}_{\phi^{P}} (Y;\bphi_{1})$ 
is monotone decreasing in $Y$ as $F_{\phi^{P}_{1}}'' < 0$ by the
part~\ref{prop:kappaconcavity-Asy} of \normalfont{Proposition}~\ref{prop:Basic-Asy}.

Considering the sum of these  terms then,
except possibly for  high enough continuation values such that
$W \in (\upbar{W},W_{gp})$
the objective function is sub-modular in
$Y$ and $\phi^{P}_{1}$ in the linear model.
Therefore, in this case by monotone comparative statics arguments
the optimal incentives $Y$ decreases in $\phi^{P}_{1}$.
This yields the result stated in the part (a) of \normalfont{Theorem}~\ref{thm:Flattening-Linears-Het} for an increase in $\phi_{1}^{P}$ in an affine increasing model
of drift ambiguity.

Next we turn to analyze the monotone comparative statics
with respect to the agentâ€™s linear component of ambiguity
ambiguity $\phi^{A}_{1}$.
Consider the partial derivative the objective function in \eqref{eqn:HJBI-Asy-Lin}
involving interaction between $Y$ and $\phi^{A}_{1}$
\[
\small{
\mathcal{H}_{\phi^{A}_{1}}(Y) =
- (1 - \phi^{P}_{1})Y
+ F_{\phi^{A}_{1}}'(W;\bphi_{1}) m(Y;\bphi_{1})
+ F'(W;\bphi_{1}) m_{\phi^{A}_{1}} (Y;\bphi_{1})
+ {F_{\phi^{A}_{1}}''(W;\bphi_{1})} r Y^{2} \sigma^{2}/2
}
\]
For monotone comparative statics argument, we characterize
monotonicity of this partial derivative in $Y$ term by term.
The first term of $\mathcal{H}_{\phi^{A}_{1}}(Y) $ is decreasing in $Y$.
Its second term
does not vary with $Y$ by envelope theorem applied to the
agent's decision problem.
Turning to the third term of $\mathcal{H}_{\phi^{A}_{1}}(Y)$,
notice that its second component
$m_{\phi^{A}_{1}} (Y;\bphi_{1}) = - a Y + \phi^{P} _{1}Y^{2} = - Y^{2}\left(1 - \phi^{A}_{1} - \phi^{P}_{1}\right)$
is decreasing in $Y$ for any viable contracting relationship:
$1 - \phi^{A}_{1} - \phi^{P}_{1} \geq 0$, for otherwise, the principal
would optimally implement zero-action.
For any viable contract,
the third term is then monotone decreasing in $Y$
for low enough continuation values such that $W \in (W_o, \upbar{W})$,
as in this range $F'(W;\bphi_{1}) >0$,
and it is decreasing in $Y$
for high enough continuation values such that
$W \in (\upbar{W},W_{gp})$ for which $F'(W;\bphi_{1}) < 0$.
Finally, the last term of $\mathcal{H}_{\phi^{A}_{1}}(Y)$
is monotone decreasing in $Y$ as $F_{\phi^{A}_{1}}'' < 0$
by the part~\ref{prop:kappaconcavity-Asy} of
\normalfont{Proposition}~\ref{prop:Basic-Asy}
Summarizing, the partial derivative $\mathcal{H}_{\phi^{A}_{1}}(Y) $
is monotone decreasing in $Y$
and hence by monotone comparative static argument
the optimal incentive sensitivity $Y$ is decreasing in $\phi^{A}_{1}$
except possibly
for high enough continuation values $W \in (\upbar{W},W_{gp})$.
This yields the result stated in the part (a) of \normalfont{Theorem}~\ref{thm:Flattening-Linears-Het} for an increase in $\phi^{A}_{1}$ in affine increasing model.


\textbf{Ambiguity linearly decreasing in action.}
Consider a linear drift ambiguity model with
$\kappa(a;\bphi^{P}) = \phi^{P}_{0} - \phi^{P}_{1} a$ and
$\kappa(a;\bphi^{A}) = \phi^{A}_{0} - \phi^{A}_{1} a$.
Denote this heterogenous ambiguity by $\bphi = (\bphi^{P},\bphi^{A})$
with $\bphi^{i} = (\phi_{0}^{i},\phi_{1}^{i})$ for $i = A,P$.

With this formulation of the drift ambiguity, the HJBI formulation~\eqref{eqn:HJBI-Asy} of the contracting problem now takes the form:
\[
\begin{split}
r F(W;\bphi) = \max_{(a,Y,c)\geq 0 }
\big\{ r \left( a - \kappa(a;\bphi^{P})  - c  \right)
+ &r   F'(W;\bphi_{1}) \left(W  - u(c)  + h(a) - \kappa(a;\bphi^{P}) Y \right) \\
&+ {F''(W;\bphi_{1})} r^{2} Y^{2} \sigma^{2}/2 \big\}
\end{split}
\]
subject to boundary conditions and
incentive compatibility \eqref{eqn:IC-Y-Asy}.

The latter implies the following relationship
between an action $a$ and incentive compatible variation in the continuation value
to the agent $Y$ required to implement it:
$Y= \frac{a}{1 - \phi^{A}_{1}}$.
As the incentive compatibility condition describes the agent's decision rule for his action choice for any given $Y$,
it depends only on linear component of
the agent's perception  $\bphi^{A}$
but not on the principal's perception $\bphi^{P}$.
Using this relationship so that $a(Y;\bphi^{A}) = (1 - \phi^{A}_{1})Y$
let $m(Y;\bphi) = \left(W - u(c) + h(a(Y;\bphi)) -\kappa(a(Y;\bphi^{A});\bphi^{P})  Y \right)$
denote the drift of the agent's continuation value (from the principal's perspective).
The latter together with the functional form for the principal's
kappa ambiguity now
now yields the HJBI equation in the following form:
\begin{equation}\label{eqn:HJBI-Asy-Lin-dec}
\begin{split}
r F(W;\bphi) = \max_{(Y,c)\geq 0 }
\big\{  r  \left(  (1 - \phi^{A}_{1})(1 - \phi^{P}_{1})Y - \phi_{0}^{P} - c  \right)
+ &r F'(W;\bphi)  m(Y;\bphi) \\
&+ \frac{F''(W;\bphi)}{2} r^{2} Y^{2} \sigma^{2} \big\}
\end{split}
\end{equation}
subject to the same set of boundary conditions.

Since the affine from of drift ambiguity are additively separable in constant and slope part,
the monotone comparative static analysis made
for the constant ambiguity model and for the linear increasing ambiguity model
 apply analogously. Specifically,
notice that in the linear decreasing formulation
constant ambiguity $\phi_{0}^{i}$  and slope ambiguity $\phi_{1}^{i}$
for each actor $i = A$ and $i=P$
enter into the formulation linearly additivity.
Due to this separability, the monotone comparative static analysis of $Y$
with respect to the constant ambiguity $\phi_{0}^{i}$
follow from the analogous arguments made for constant ambiguity
model in the earlier part of this proof.

Similarly the monotone comparative static analysis with respect to $\phi_{1}^{i}$
for each actor $i = A$ and $i=P$
follows from analogous arguments made for
linear increasing model by a change of variable.
As in the monotone comparative static analysis of the model with symmetric ambiguity,
this change of variable replaces
$\phi_{1}^{i}$ in the linear increasing ambiguity model
with $-\phi_{1}^{i}$ to represent drift set shrinking as effort increases.
In this specification, higher drift ambiguity
arises from an increase in $\phi_{0}^{i}$ and a decrease in $\phi_{1}^{i}$,
which expand drift sets.
From this change of variables, the analogous characterization follows.
for an affine decreasing model of drift ambiguity.

\bibliographystyle{chicago}
\bibliography{PAAmbigLit,imp-mh-ambig}
\end{document}